\documentclass[12pt]{article}
\usepackage[english]{babel}
\usepackage{tikz}\usetikzlibrary{matrix}
\usepackage{amssymb}
\usepackage[T1]{fontenc}
\usepackage[utf8]{inputenc}
\usepackage{calrsfs}
\usepackage{amsmath}
\usepackage{authblk}
\usepackage{amsthm}
\usepackage{enumerate}
\usepackage{caption}
\usepackage{a4wide}
\usetikzlibrary[shapes.misc]
\usetikzlibrary[shapes.geometric]
\newtheorem{de}{Definition}[section]
\newtheorem{theo}{Theorem}[section]
\newtheorem{prop}[theo]{Proposition}
\newtheorem{cor}[theo]{Corollary}
\newtheorem{lem}[theo]{Lemma}

\newtheorem{con}{Conjecture}

\newtheorem{quest}[theo]{Question}

\title{On $S$-packing edge-colorings of cubic graphs}

\author[1,2]{Nicolas Gastineau}
\author[1]{Olivier Togni}

\affil[1]{LE2I FRE2005, CNRS, Arts et Métiers, Université Bourgogne Franche-Comté, F-21000 Dijon, France}
\affil[2]{PSL, Université Paris-Dauphine, LAMSADE UMR CNRS 7243, France}

\begin{document}
\maketitle
\begin{abstract}
Given a non-decreasing sequence $S=(s_1,s_2, \ldots, s_k)$ of positive integers, an {\em $S$-packing edge-coloring} of a graph $G$ is a partition of the edge set of $G$ into $k$ subsets $\{X_1,X_2, \ldots, X_k\}$ such that for each  $1\le i\le k$, the distance between two distinct edges $e,e'\in X_i$ is at least $s_i+1$. This paper studies $S$-packing edge-colorings of cubic graphs. Among other results, we prove that cubic graphs having a $2$-factor are $(1,1,1,3,3)$-packing edge-colorable, $(1,1,1,4,4,4,4,4)$-packing edge-colorable and $(1,1,2,2,2,2,2)$-packing edge-colorable. We determine sharper results for cubic graphs of bounded oddness and $3$-edge-colorable cubic graphs and we propose many open problems.
\end{abstract}

\section{Introduction}
All the graphs considered in this paper are simple and connected, unless stated otherwise.
A proper {\em edge-coloring} of a graph $G$ is a mapping which associates a color (an integer) to each edge such that adjacent edges get distinct colors. In such a coloring, each color class is a matching (also called stable set of edges or 1-packing). 
According to Vizing's famous theorem, every cubic graph needs either 3 or 4 colors for a proper edge-coloring. The bridgeless cubic graphs (often with other restrictions) which are not edge-colorable with three colors are called {\em snarks}~\cite{BGHM,Luk}. 

As an extension, a $d$-strong edge-coloring of $G$ is a proper coloring such that edges at distance at most $d$ have distinct colors, i.e., a partition of $E(G)$ into sets of edges at pairwise distance at least $d+1$, also called $d$-packings. A $2$-strong edge-coloring is simply called a strong edge-coloring and a $2$-packing of edges is an induced matching.
Strong edge-colorings of cubic graphs retain a lot of attention since decades~\cite{And,Fou,Hoc}.

The aim of this paper is to study a mixing of these two types of edge-colorings, i.e., colorings of (sub)cubic graphs in which some color classes are 1-packings while other are $d$-packings, $d\ge 2$. More formally, given a non-decreasing sequence $S=(s_1,s_2, \ldots, s_k)$ of positive integers, an {\em $S$-packing edge-coloring} of a graph $G$ is a partition of the edge set of $G$ into $k$ subsets $\{X_1,X_2, \ldots, X_k\}$ such that each $X_i$ is an $s_i$-packing, $1\le i\le k$.

The vertex analogous of $S$-packing edge-coloring has been first studied by Goddard and Xu~\cite{GX1,GX2} and then recently on cubic graphs~\cite{Bal17,BresKlav,BKRW,Gasto,GHT}. The particular case of $(1,2,\ldots,k)$-packing coloring has been the subject of many papers (see~\cite{FiCo,Fin,nico}) since its introduction by Goddard et al.~\cite{God}.

For an edge-coloring, a color for which the color class is an $r$-packing is said to be a color of {\em radius} $r$.
In order to avoid long subsequences of the same integer in sequences of colors, we sometimes use the exponent to denote repetitions of an integer, e.g., $(1^2,2^5)=(1,1,2,2,2,2,2)$.
 Also, to simplify, an $S$-packing edge-coloring will be simply called an {\em $S$-coloring} in the remainder of the paper. A $(1,1,1,2)$-coloring and a $(1,1,2,2,2)$-coloring of the Petersen graph are illustrated in Figure \ref{figmp} (one can check that the Petersen graph is not $(1,1,2,2)$-colorable).

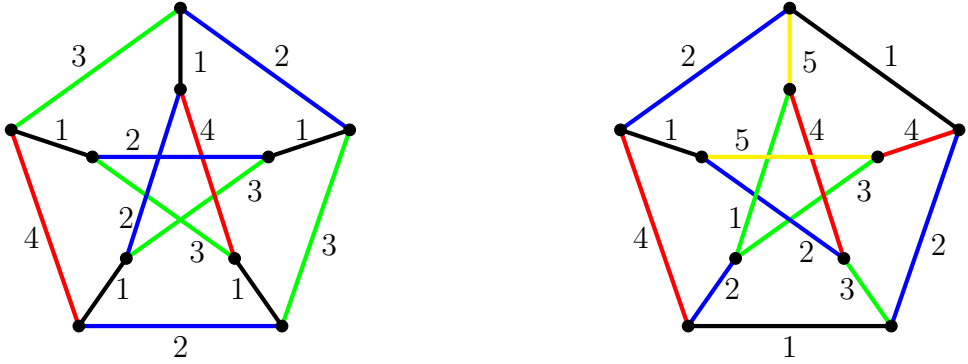
\begin{figure}[t]
\begin{center}
\begin{tikzpicture}[scale=1.35]
\draw[ultra thick,color=blue] (0,-0.2) -- (3*2/3,-0.2);
\draw[ultra thick, color=red] (0,-0.2) -- (-1*2/3,2.6*2/3);
\draw[ultra thick,color=green] (3*2/3,-0.2) -- (4*2/3,2.6*2/3);
\draw[ultra thick,color=green] (-1*2/3,2.6*2/3) -- (1.5*2/3,4.4*2/3);
\draw[ultra thick,color=blue] (4*2/3, 2.6*2/3) -- (1.5*2/3,4.4*2/3);
\draw[ultra thick,color=black] (1.5*2/3,3.2*2/3) -- (1.5*2/3,4.4*2/3);
\draw[ultra thick,color=black] (-1*2/3,2.6*2/3) -- (0.2*2/3, 2.2*2/3);
\draw[ultra thick,color=black] (4*2/3, 2.6*2/3) -- (2.8*2/3, 2.2*2/3);
\draw[ultra thick,color=black] (0,-0.2) -- (0.7*2/3,0.7*2/3);
\draw[ultra thick,color=black] (3*2/3,-0.2) -- (2.3*2/3,0.7*2/3);
\draw[ultra thick,color=green]  (0.7*2/3,0.7*2/3) -- (2.8*2/3,2.2*2/3);
\draw[ultra thick,color=green]   (2.3*2/3,0.7*2/3) -- (0.2*2/3,2.2*2/3);
\draw[ultra thick,color=blue] (0.7*2/3,0.7*2/3) -- (1.5*2/3,3.2*2/3);
\draw[ultra thick, color=red]  (2.3*2/3,0.7*2/3) -- (1.5*2/3,3.2*2/3);
\draw[ultra thick,color=blue]  (0.2*2/3,2.2*2/3) -- (2.8*2/3,2.2*2/3);
\node at (0,-0.2) [ circle,fill=black,draw=black,scale=0.4] {};
\node at (-1*2/3,2.6*2/3) [ circle,fill=black,draw=black,scale=0.4] {};
\node at (3*2/3,-0.2) [ circle,fill=black,draw=black,scale=0.4] {};
\node at (4*2/3,2.6*2/3)[ circle,fill=black,draw=black,scale=0.4] {};
\node at (1.5*2/3,4.4*2/3) [ circle,fill=black,draw=black,scale=0.4] {};
\node at (1.5*2/3,3.2*2/3)[ circle,fill=black,draw=black,scale=0.4] {};
\node at (0.7*2/3,0.7*2/3) [ circle,fill=black,draw=black,scale=0.4] {};
\node at (2.8*2/3,2.2*2/3) [ circle,fill=black,draw=black,scale=0.4] {};
\node at (0.2*2/3,2.2*2/3) [ circle,fill=black,draw=black,scale=0.4] {};
\node at (2.3*2/3,0.7*2/3)[ circle,fill=black,draw=black,scale=0.4] {};
\node at (1.5*2/3,-0.4){$2$};
\node at (3.7*2/3,0.9*2/3){$3$};
\node at (-0.7*2/3,1*2/3){$4$};
\node at (0*2/3,3.7*2/3){$3$};
\node at (3*2/3,3.7*2/3){$2$};
\node at (0.65*2/3,0.25*2/3){$1$};
\node at (2.35*2/3,0.25*2/3){$1$};
\node at (-0.25*2/3,2.6*2/3){$1$};
\node at (3.3*2/3,2.6*2/3){$1$};
\node at (1.8*2/3,3.6*2/3){$1$};
\node at (.8*2/3,2.45*2/3){$2$};
\node at (1.9*2/3,2.6*2/3){$4$};
\node at (0.7*2/3,1.3*2/3){$2$};
\node at (1.75*2/3,.8*2/3){$3$};
\node at (2.6*2/3,1.7*2/3){$3$};

\draw[ultra thick,color=black] (0+6,-0.2) -- (3*2/3+6,-0.2);
\draw[ultra thick, color=red] (0+6,-0.2) -- (-1*2/3+6,2.6*2/3);
\draw[ultra thick,color=blue] (3*2/3+6,-0.2) -- (4*2/3+6,2.6*2/3);
\draw[ultra thick,color=blue] (-1*2/3+6,2.6*2/3) -- (1.5*2/3+6,4.4*2/3);
\draw[ultra thick,color=black] (4*2/3+6, 2.6*2/3) -- (1.5*2/3+6,4.4*2/3);
\draw[ultra thick,color=yellow] (1.5*2/3+6,3.2*2/3) -- (1.5*2/3+6,4.4*2/3);
\draw[ultra thick,color=black] (-1*2/3+6,2.6*2/3) -- (0.2*2/3+6, 2.2*2/3);
\draw[ultra thick,color=red] (4*2/3+6, 2.6*2/3) -- (2.8*2/3+6, 2.2*2/3);
\draw[ultra thick,color=blue] (6,-0.2) -- (0.7*2/3+6,0.7*2/3);
\draw[ultra thick,color=green] (3*2/3+6,-0.2) -- (2.3*2/3+6,0.7*2/3);
\draw[ultra thick,color=green]  (0.7*2/3+6,0.7*2/3) -- (2.8*2/3+6,2.2*2/3);
\draw[ultra thick,color=blue]   (2.3*2/3+6,0.7*2/3) -- (0.2*2/3+6,2.2*2/3);
\draw[ultra thick,color=green] (0.7*2/3+6,0.7*2/3) -- (1.5*2/3+6,3.2*2/3);
\draw[ultra thick, color=red]  (2.3*2/3+6,0.7*2/3) -- (1.5*2/3+6,3.2*2/3);
\draw[ultra thick,color=yellow]  (0.2*2/3+6,2.2*2/3) -- (2.8*2/3+6,2.2*2/3);
\node at (6,-0.2) [ circle,fill=black,draw=black,scale=0.4] {};
\node at (-1*2/3+6,2.6*2/3) [ circle,fill=black,draw=black,scale=0.4] {};
\node at (3*2/3+6,-0.2) [ circle,fill=black,draw=black,scale=0.4] {};
\node at (4*2/3+6,2.6*2/3)[ circle,fill=black,draw=black,scale=0.4] {};
\node at (1.5*2/3+6,4.4*2/3) [ circle,fill=black,draw=black,scale=0.4] {};
\node at (1.5*2/3+6,3.2*2/3)[ circle,fill=black,draw=black,scale=0.4] {};
\node at (0.7*2/3+6,0.7*2/3) [ circle,fill=black,draw=black,scale=0.4] {};
\node at (2.8*2/3+6,2.2*2/3) [ circle,fill=black,draw=black,scale=0.4] {};
\node at (0.2*2/3+6,2.2*2/3) [ circle,fill=black,draw=black,scale=0.4] {};
\node at (2.3*2/3+6,0.7*2/3)[ circle,fill=black,draw=black,scale=0.4] {};
\node at (1.5*2/3+6,-0.4){$1$};
\node at (3.7*2/3+6,0.9*2/3){$2$};
\node at (-0.7*2/3+6,1*2/3){$4$};
\node at (0*2/3+6,3.7*2/3){$2$};
\node at (3*2/3+6,3.7*2/3){$1$};
\node at (0.65*2/3+6,0.25*2/3){$2$};
\node at (2.35*2/3+6,0.25*2/3){$3$};
\node at (-0.25*2/3+6,2.6*2/3){$1$};
\node at (3.3*2/3+6,2.6*2/3){$4$};
\node at (1.8*2/3+6,3.6*2/3){$5$};
\node at (.8*2/3+6,2.45*2/3){$5$};
\node at (1.9*2/3+6,2.6*2/3){$4$};
\node at (0.7*2/3+6,1.3*2/3){$1$};
\node at (1.75*2/3+6,.8*2/3){$2$};
\node at (2.6*2/3+6,1.7*2/3){$3$};

\end{tikzpicture}
\end{center}
\caption{A $(1,1,1,2)$-coloring (on the left, with colors 1, 2 and 3 of radius 1 and color 4 of radius $2$) and a $(1,1,2,2,2)$-coloring (on the right, with colors 1 and 2 of radius 1 and colors 4 and 5 of radius $2$) of the Petersen graph.}
\label{figmp}
\end{figure}

Let $G$ be a graph and $A\subseteq E(G)$. By $G^{k}[A]$, we denote the graph with vertex set $A$ and edge set $\{ee'\in E(G) |\ e\in A,\ e'\in A,\ d_G(e,e')\le k\}$, where $d_G(e,e')$ is the usual distance between the two edges $e$ and $e'$ in $G$.
We recall that a {\em $2$-factor} of $G$ is a spanning subgraph of $G$ that consists in a disjoint union of cycles.
For a cubic graph $G$ having a $2$-factor, the {\em oddness} of $G$ is the minimum number of odd cycle among all $2$-factors of $G$. According to Petersen's theorem, every bridgeless cubic graph has a $2$-factor.

\begin{de}
For a graph $G$ with a $2$-factor $\mathcal{F}$ and a set $A\subseteq E(\mathcal{F})$, we use the following notation:
\begin{enumerate}
\item[i)] $A$ is of type I if it contains exactly one edge per odd cycle of $\mathcal{F}$ and no edge of any even cycle of $\mathcal{F}$;
\item[ii)] $A$ is of type II  if no two edges of $A$ are adjacent in $G$ and if $A$ contains $\lfloor n/2 \rfloor$ edges in every cycle of length $n$ from $\mathcal{F}$, $n\ge 3$.
\end{enumerate}
\end{de}
These two definitions will be used several times in the paper in order to describe the edges that remain to be colored in a cubic graph in which a maximum number of edges are colored with one, two or three colors of radius $1$. Notice that a set of type I is also called an odd cycle (edge) transversal of $\mathcal{F}$.

\begin{table}[h]
   \caption{\label{sum} The minimum integer $n=\ell+m$ in order that all cubic graphs (and all $3$-edge-colorable cubic graphs) having a $2$-factor are $(1^\ell,k^{m})$-colorable (the bold numbers represent the exact values of $n$ and a pair of two integers $a$-$b$ represents a lower bound and an upper bound on $n$).}
\begin{center}
\begin{tabular}{|c|c|c|c|c|c|c|}
  \hline
\ Class \ & \multicolumn{3}{c|}{cubic graphs}  &  \multicolumn{3}{c|}{3-edge-colorable}  \\
  & \multicolumn{3}{c|}{}  &  \multicolumn{3}{c|}{cubic graphs}  \\
\hline
  $k \backslash \ell$ & \ \ \ \ \ \ 1 \ \ \ \ \ \ & \ \ \ \ \ \ 2 \ \ \ \ \ \ & \ \ \ \ \ \ 3 \ \ \ \ \ \ &  \ \ \ \ \ \ 1 \ \ \ \ \ \ & \ \ \ \ \ \ 2 \ \ \ \ \ \ & \ \ \ \ \ \ 3 \ \ \ \ \ \ \\
  \hline
  2 & 8-10 & 6-7 & \ \ \ \ \ \textbf{4}~\cite{Fou2}& 6-9 & 5-6 & \textbf{3}  \\
  3 & 15-21 & 9-13 & \textbf{5} & 15-19 & 9-11 & \textbf{3}  \\
  4 & 31-48  & 17-28 & 5-8 & 31-43 & 17-23 & \textbf{3}  \\
  \hline
\end{tabular}
\end{center}
\end{table}
As any subcubic graph $H$ is the subgraph of a cubic graph $G$ and as $d_H(e,e')\ge d_G(e,e')$ for any two edges $e,e'\in E(H)$, then any $S$-coloring of $G$ is also an $S$-coloring of $H$. Therefore, the results of this paper that are not concerned with oddness can be easily extended to subcubic graphs.

Table \ref{sum} summarizes the main results proven in this paper. Note that the lower bounds, except for the sequence $(1,1,1,2)$ and the sequences $(1,1,1,\ldots)$ for $3$-edge-colorable graphs, have been determined by computer.

The paper is organized as follows. In section 2, we begin by presenting structural results about sets of type I and the relation between $S$-coloring and sets of type I and II.
We prove in Section 3 that cubic graphs having a $2$-factor are $(1,1,1,3,3)$-colorable and conjecture that all cubic graphs are $(1,1,1,3)$-colorable, except the Petersen and Tietze graphs. This conjecture is proven for some restricted classes of snarks. Similar results are given for sequences of type $(1,1,1,4,\ldots,4)$. In Section 4 we study $(1,1,k,\ldots, k)$-colorings and prove that all cubics graphs having a $2$-factor are $(1^2,2^5)$-colorable and also colorable with two colors of radius one and a finite number of colors of radius $k$, for any $k\ge 2$. In Section 5 we prove that for a fixed integer $k$, every cubic graph having a $2$-factor is $(1,k,\ldots, k)$-colorable with a finite number of occurrences of $k$ in the sequence. Finally, in Section 6 we prove that for every positive integer $k$, there exists a subcubic graph which is not $(1,2,\ldots,k)$-colorable.

\section{Sets of type I and II}
\begin{figure}[t]
\begin{center}
\begin{tikzpicture}[scale=1.15]
\draw[ultra thick, color=blue] (0,-0.2) -- (3*2/3,-0.2);
\draw[ultra thick,color=blue] (0,-0.2) -- (-1*2/3,2.6*2/3);
\draw[ultra thick, dashed,color=blue] (3*2/3,-0.2) -- (4*2/3,2.6*2/3);
\draw[ultra thick,color=blue] (-1*2/3,2.6*2/3) -- (1.5*2/3,4.4*2/3);
\draw[ultra thick, color=blue] (4*2/3, 2.6*2/3) -- (1.5*2/3,4.4*2/3);
\draw[color=black] (1.5*2/3,3.2*2/3) -- (1.5*2/3,4.4*2/3);
\draw[color=black] (-1*2/3,2.6*2/3) -- (0.2*2/3, 2.2*2/3);
\draw[color=black] (0,-0.2) -- (0.7*2/3,0.7*2/3);
\draw[ultra thick,color=blue]  (0.7*2/3,0.7*2/3) -- (2.8*2/3,2.2*2/3);
\draw[ultra thick,color=blue]   (2.3*2/3,0.7*2/3) -- (0.2*2/3,2.2*2/3);
\draw[ultra thick,color=blue] (0.7*2/3,0.7*2/3) -- (1.5*2/3,3.2*2/3);
\draw[ultra thick,color=blue]  (2.3*2/3,0.7*2/3) -- (1.5*2/3,3.2*2/3);
\draw[ultra thick,color=blue]  (0.2*2/3,2.2*2/3) -- (2.8*2/3,2.2*2/3);
\node at (0,-0.2) [ circle,fill=black,draw=black,scale=0.4] {};
\node at (-1*2/3,2.6*2/3) [ circle,fill=black,draw=black,scale=0.4] {};
\node at (3*2/3,-0.2) [ circle,fill=black,draw=black,scale=0.4] {};
\node at (4*2/3,2.6*2/3)[ circle,fill=black,draw=black,scale=0.4] {};
\node at (1.5*2/3,4.4*2/3) [ circle,fill=black,draw=black,scale=0.4] {};
\node at (1.5*2/3,3.2*2/3)[ circle,fill=black,draw=black,scale=0.4] {};
\node at (0.7*2/3,0.7*2/3) [ circle,fill=black,draw=black,scale=0.4] {};
\node at (2.8*2/3,2.2*2/3) [ circle,fill=black,draw=black,scale=0.4] {};
\node at (0.2*2/3,2.2*2/3) [ circle,fill=black,draw=black,scale=0.4] {};
\node at (2.3*2/3,0.7*2/3)[ circle,fill=black,draw=black,scale=0.4] {};

\draw[ultra thick,color=blue] (0+5,-0.2) -- (3*2/3+5,-0.2);
\draw[ultra thick, color=blue] (0+5,-0.2) -- (-1*2/3+5,2.6*2/3);
\draw[ultra thick,color=blue] (3*2/3+5,-0.2) -- (4*2/3+5,2.6*2/3);
\draw[ultra thick,dashed,color=blue] (-1*2/3+5,2.6*2/3) -- (1.5*2/3+5,4.4*2/3);
\draw[ultra thick,color=blue] (4*2/3+5, 2.6*2/3) -- (1.5*2/3+5,4.4*2/3);
\draw[color=black] (1.5*2/3+5,3.2*2/3) -- (1.5*2/3+5,4.4*2/3);
\draw[color=black] (4*2/3+5, 2.6*2/3) -- (2.8*2/3+5, 2.2*2/3);
\draw[color=black] (3*2/3+5,-0.2) -- (2.3*2/3+5,0.7*2/3);
\draw[ultra thick,color=blue]  (0.7*2/3+5,0.7*2/3) -- (2.8*2/3+5,2.2*2/3);
\draw[ultra thick,color=blue]   (2.3*2/3+5,0.7*2/3) -- (0.2*2/3+5,2.2*2/3);
\draw[ultra thick,dashed,color=blue] (0.7*2/3+5,0.7*2/3) -- (1.5*2/3+5,3.2*2/3);
\draw[ultra thick,color=blue]  (2.3*2/3+5,0.7*2/3) -- (1.5*2/3+5,3.2*2/3);
\draw[ultra thick,color=blue]  (0.2*2/3+5,2.2*2/3) -- (2.8*2/3+5,2.2*2/3);
\node at (0+5,-0.2) [ circle,fill=black,draw=black,scale=0.4] {};
\node at (-1*2/3+5,2.6*2/3) [ circle,fill=black,draw=black,scale=0.4] {};
\node at (3*2/3+5,-0.2) [ circle,fill=black,draw=black,scale=0.4] {};
\node at (4*2/3+5,2.6*2/3)[ circle,fill=black,draw=black,scale=0.4] {};
\node at (1.5*2/3+5,4.4*2/3) [ circle,fill=black,draw=black,scale=0.4] {};
\node at (1.5*2/3+5,3.2*2/3)[ circle,fill=black,draw=black,scale=0.4] {};
\node at (0.7*2/3+5,0.7*2/3) [ circle,fill=black,draw=black,scale=0.4] {};
\node at (2.8*2/3+5,2.2*2/3) [ circle,fill=black,draw=black,scale=0.4] {};
\node at (0.2*2/3+5,2.2*2/3) [ circle,fill=black,draw=black,scale=0.4] {};
\node at (2.3*2/3+5,0.7*2/3)[ circle,fill=black,draw=black,scale=0.4] {};

\draw[dashed,color=black] (4*2/3, 2.6*2/3) -- (-1*2/3+5,2.6*2/3);
\draw[dashed,color=black] (0.7*2/3+5,0.7*2/3) --(3*2/3,-0.2);
\draw[color=black] (0.2*2/3+5, 2.2*2/3) --(2.8*2/3, 2.2*2/3);
\draw[color=black] (0+5,-0.2) --(2.3*2/3,0.7*2/3);
\node at (3.9*2/3,1.2*2/3) {$e_{1}$};
\node at (0.1*2/3+5,3.8*2/3) {$e_{2}$};
\node at (0.9*2/3+5,2.4*2/3) {$e_{3}$};

\end{tikzpicture}
\end{center}
\caption{A 2-factor containing three edges $e_{1}$, $e_2$ and $e_3$ being in different odd cycles which induce a subgraph containing a path of length $5$ (dashed lines: edges in the path of length $5$, thick lines: edges from the 2-factor).}
\label{figmp}
\end{figure}
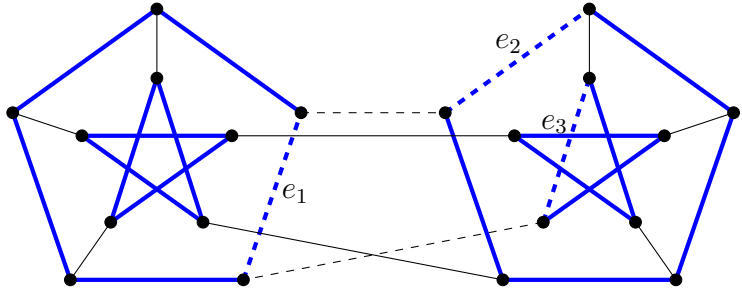
The following result from Fouquet and Vanherpe is a structural property about non 3-edge-colorable cubic graphs \cite{Fou2}. It will be used several times in the proofs of Section 3.

\begin{prop}[\cite{Fou2}]\label{fouqq}
Let $G$ be a cubic graph having a $2$-factor. Let $\mathcal{F}$ be a $2$-factor of $G$ containing a minimum number of odd cycles and let $\mathcal{F'}$ be the set of odd cycles from $\mathcal{F}$. Then, no three edges being in different cycles of $\mathcal{F}'$ induce in $G$ a subgraph containing a path of length $5$.
\end{prop}
Figure \ref{figmp} illustrates three edges in the forbidden configuration described in Proposition~\ref{fouqq} in a fixed 2-factor of a cubic graph (by Proposition~\ref{fouqq}, it means that there exists a $2$-factor containing at most two odd cycles in the graph from Figure \ref{figmp}).
Note that the previous proposition was not explicitly presented in the paper of Fouquet and Vanherpe but can be easily obtained by combining Properties 3 and 6 of~\cite[Theorem 4]{Fou2}. Note also that the proof of Theorem 4 is written in another paper~\cite{Fou3} from the same authors. 

\begin{lem}[\cite{Fou2}]\label{fouqqut}
Let $G$ be a cubic graph having a $2$-factor. Let $\mathcal{F}$ be a $2$-factor of $G$ containing a minimum number of odd cycles and let $\mathcal{F'}$ be the set of odd cycles from $\mathcal{F}$. There exists a set $A$ of type I in $\mathcal{F'}$ such that $G^{2}[A]$ is an empty graph, i.e., $\chi(G^{2}[A])\le 1$.
\end{lem}
As the previous proposition, Lemma~\ref{fouqqut} has not been explicitly presented in the paper of Fouquet and Vanherpe but is an intermediate step to prove a Theorem~\cite[Theorem 8]{Fou2}.

The following lemma will allow in Sections 3, 4 and 5, to reduce the problem of finding an $S$-coloring to the one of finding a set $A$ of type I and/or II such that $G^k[A]$ has small chromatic number (where $2\le k\le 4$ and $k$ appears in $S$).

\begin{lem}\label{lem1111}
Let $G$ be a cubic graph having a $2$-factor $\mathcal{F}$ and let $\ell$, $\ell'$ and $\ell''$ be positive integers. Let $A\subseteq E(\mathcal{F})$, $B\subseteq E(\mathcal{F})$ and $C\subseteq E(\mathcal{F})$ be sets such that $A$, $B$ and $C$ form a partition of $E(\mathcal{F})$, $A$ being of type I, $B$ and $C$ being of type II. The following properties hold:
\begin{enumerate}
\item[i)] if $\chi(G^{k}[A])\le \ell$, then $G$ is $(1,1,1,k^{\ell})$-colorable;
\item[ii)] if $\chi(G^{k}[A])\le \ell$ and $\chi(G^{k}[B])\le \ell'$, then $G$ is $(1,1,k^{\ell+\ell'})$-colorable;
\item[iii)] if $\chi(G^{k}[A])\le \ell$, $\chi(G^{k}[B])\le \ell'$ and $\chi(G^{k}[C])\le \ell''$, then $G$ is $(1,k^{\ell+\ell'+\ell''})$-colorable.
\end{enumerate}
\end{lem}
\begin{proof}
Note that $G-\mathcal{F}$ is a matching and thus can be colored with a color of radius $1$. Note also that $B$ and $C$ are matchings and thus can be also colored each with a color of radius $1$.\newline
i) We color $E(G-\mathcal{F})$, $B$ and $C$ with three colors of radius $1$.
Then, any $\ell$-coloring of $G^{k}[A]$ induces a coloring of the edges of $A$ with the remaining colors, i.e., the $\ell$ colors of radius $k$.\newline
ii) We color $E(G-\mathcal{F})$ and $C$ with two colors of radius $1$. Then, any $\ell$-coloring of $G^{k}[A]$ ($\ell'$-coloring of $G^{k}[B]$, respectively) induces a coloring of the edges of $A$ ($B$, respectively) with $\ell$ ($\ell'$, respectively) colors of radius $k$.\newline
iii) We color $E(G-\mathcal{F})$ with one color of radius $1$. Then, any $\ell$-coloring of $G^{k}[A]$ ($\ell'$-coloring of $G^{k}[B]$, $\ell''$-coloring of $G^k[C]$, respectively) induces a coloring of the edges of $A$ ($B$, $C$, respectively) with $\ell$ ($\ell'$, $\ell''$, respectively) colors of radius $k$.
\end{proof}
\section{$(1,1,1,k,\ldots,k)$-coloring}

In this section, we first prove a general upper bound on the required number of colors of radius $k$ in order that all cubic graphs having a $2$-factor are $(1,1,1,k,\ldots,k)$-colorable. We also prove that all cubic graphs having a $2$-factor are $(1,1,1,3,3)$-colorable and that some cubic graphs having a $2$-factor are $(1,1,1,3)$-colorable.

For the case $k=2$, Payan~\cite{Pay} has shown that one color of radius two is sufficient. Another proof of this result has been given by Fouquet and Vanherpe~\cite{Fou2} (Lemma~\ref{fouqqut} is an intermediate step of the proof of this result). 
\begin{theo}[\cite{Fou2, Pay}]\label{theo1112}
Every cubic graph is $(1,1,1,2)$-colorable.
\end{theo}
Notice that this result is tight since the Petersen graph is not $(1,1,1,3)$-colorable (as it is not 3-edge-colorable and has diameter 2).
In order to find similar results for $k\ge 3$, we consider the sequence of integers $(a_k)_{k\ge2}$ defined by: $a_2=2$, $a_3=4$ and $a_k= a_{k-1}+2 a_{k-2}+2$, for $k\ge 4$. Note that this sequence is contained in Sloane Online Encyclopedia of Integers Sequences (A026644 and A167030) and that $a_k=\frac{2^{k+1}-(-1)^{k+1}-3}{3}$ for $k\ge 2$.

\begin{lem}\label{lem1}
Let  $G$ be cubic graph having a $2$-factor, let $\mathcal{F}$ be any $2$-factor of $G$ and let $A\subseteq E(\mathcal{F})$ be of type I.
For $k\ge 2$, the graph $G^{k}[A]$, satisfies $\Delta (G^{k}[A])\le a_k$.
\end{lem}
\begin{proof}
Let $\mathcal{F}'$ be the subgraph of $\mathcal{F}$ only containing the odd cycles of $\mathcal{F}$. Note that $A\subseteq E(\mathcal{F}')$.
Let $e$ be an edge of $\mathcal{F}$ and let $D^k_e$ be the subgraph of $\mathcal{F}'$ induced by the vertices at distance at most $k-1$ from any extremity of $e$. For an integer $k$, let $n_k$ be the maximum number of connected components of $D^k_e$ among all cubic graphs having a $2$-factor, all choice of $2$-factor in these graphs and all choice of edge $e$ in these $2$-factors. 
In order to prove that $\Delta (G^{k}[A])\le a_k$, we are going to show that $n_k\le a_k+1$.
Note that several connected components of $D^k_e$ can be included in the same odd cycle of $\mathcal{F}'$. 

The number of connected components of $D^{k+1}_e$ which are not in $D^{k}_e$ is bounded by two times the number of connected components of $D^k_e$ which are not trivial (the connected components containing only one vertex) which is itself bounded by $2n_{k-1}$.
We can easily prove that $n_2\le3=a_2+1$ and $n_3\le 5=a_3+1$. By induction, we suppose that $n_{k-1}\le a_{k-1}+1$ and $n_k\le a_k+1$.
For $k\ge 2$, we obtain that $n_{k+1}\le n_k+2n_{k-1} \le a_k+1+2 (a_{k-1}+1)\le a_k +2 a_{k-1} +3\le a_{k+1}+1$.

Consequently, since each connected component of $D^{k}_e$ contains either at most one vertex incident with an edge of $A$ or at most two vertices incident with an edge of $A$ in the case this edge lies in $D^k_e$, we have $\Delta (G^{k}[A])\le n_{k}-1\le a_{k}$ (the minus 1 comes from the connected component containing $e$).
\end{proof}

\subsection{$(1,1,1,3,\ldots,3)$-coloring}
In this subsection, we try to minimize the number of required integers 3 in order that all cubic graphs having a $2$-factor are $(1,1,1,3,\ldots,3)$-colorable. We give two results about this problem but we are aware that the result of the following theorem can probably be sharpened since we have not been able to find an infinite family of non $(1,1,1,3)$-colorable cubic graphs. However, at the end of the subsection, we exhibit two non $(1,1,1,3)$-colorable bridgeless cubic graphs.

\begin{theo}\label{2*3}
Every cubic graph having a $2$-factor is $(1,1,1,3,3)$-colorable.
\end{theo}
\begin{proof}
Let $G$ be a cubic graph. Let $\mathcal{F}$ be a $2$-factor of $G$ having a minimum number of odd cycles. Let $\mathcal{F}'$ be the set of odd cycles from $\mathcal{F}$. By Lemma \ref{lem1111}.i), if there exists a set $A$ of type I in $\mathcal{F'}$ such that $\Delta(G^{3}[A])\le 1$, then $G$ is $(1,1,1,3,3)$-colorable.

We give labels to the vertices of $\mathcal{F}'$ as follows. If a vertex belonging to an odd cycle $C$ from $\mathcal{F}'$ has a neighbor in a different cycle of $\mathcal{F}'$, we label it by $+$, otherwise we label it by $-$. 
By Proposition \ref{fouqq}, the two end vertices of any edge of any cycle $C$ of $\mathcal{F'}$ have neighbors in only at most one cycle of $\mathcal{F'}$ other than $C$. Thus, if consecutive vertices are labeled by $+$ in an odd cycle $C$ from $\mathcal{F}'$, then there exists an unique cycle $C'$ of $\mathcal{F}'$ such that all these vertices have neighbors only in $C\cup C'$. 

\vspace{0.5cm}
\textbf{Observation 1.} For each edge $e\in E(G-\mathcal{F})$ having an extremity $u$ in a odd cycle $C$ from $\mathcal{F}'$ the following is true:
\begin{enumerate}
\item[i)] if $u$ is labeled by $+$, then all edges in $\mathcal{F}'$ at distance at most $2$ from $e$ are included in $C\cup C'$, for $C'$ an odd cycle of $\mathcal{F}'$;
\item[ii)] if $u$ is labeled by $-$, then all edges in $\mathcal{F}'$ at distance at most $2$ from $e$ are included in $C$.
\end{enumerate}

The previous observation can be easily obtained using the fact that $e$ is only adjacent with edges of $\mathcal{F}$, these edges being themselves adjacent with edges either in the same cycle than they or in $G-\mathcal{F}$.

We will construct $A$, starting from an empty set, as follows. Since each cycle $C$ of $\mathcal{F}'$ has an odd number of vertices there exist two consecutive vertices both labeled either by $+$ or by $-$ in every cycle of $\mathcal{F}'$.
Let $u_1$ and $u_{2}$ be these two adjacent vertices (both labeled either by $+$ or $-$) and suppose that $u_{0}$, $u_{1}$, $u_{2}$, $u_{3}$ and $u_{4}$ are consecutive vertices of the cycle $C$ (if $C$ contains three vertices then $u_{3}=u_{0}$ and $u_4=u_{1}$). For each cycle $C$ of $\mathcal{F}'$, we add to $A$ an edge of $C$ depending on the label of $u_{1}$ and $u_{2}$.

If $u_{1}$ and $u_{2}$ are both labeled by $+$ in $C$, then we add the edge $u_{1}u_{2}$ into $A$. Note that $u_{0}, u_{1}, u_{2}, u_{3}$ are labeled either by $+,+,+,+$, by $-,+,+,+$, by $+,+,+,-$ or by $-,+,+,-$ (the labels are given following the index of $u$).
Consequently, by Observation 1, there exists a cycle $C'$ from $\mathcal{F}'$ such that all edges from $\mathcal{F}'$ at distance at most $3$ from $u_{1} u_{2}$ are in $C\cup C'$.

If $u_{1}$ and $u_{2}$ are both labeled by $-$, then we add the edge $u_{2} u_{3}$ to $A$. Note that $u_{1}, u_{2}, u_{3}, u_{4}$ are labeled either by $-,-,+,+$, by $-,-,-,+$, by $-,-,+,-$ or by $-,-,-,-$. Also in this case, by Observation 1, there exists a cycle $C'$ from $\mathcal{F}'$ such that all edges from $\mathcal{F}'$ at distance at most $3$ from $u_{1} u_{2}$ are in $C\cup C'$.

Since there is one edge of $A$ per cycle of $\mathcal{F}'$ we obtain, by construction, that $\Delta(G^{3}[A])\le 1$ and thus $\chi(G^{3}[A])\le 2$.
Finally, by Lemma \ref{lem1111}.i), $G$ is $(1,1,1,3,3)$-colorable.

\end{proof}

\begin{figure}[t]
\begin{center}
\begin{tikzpicture}[scale=1.1]
\draw[ultra thick,color=blue] (0,2) -- (2*0.464723,2*0.88545);
\draw[ultra thick,color=blue] (2*0.464723,2*0.88545) -- (2*0.822983,2*0.56806);
\draw (2*0.822983,2*0.56806) -- (2*0.992708,2*0.12053);
\draw[ultra thick,color=blue] (2*0.992708,2*0.12053) -- (2*0.935016,2*-0.35460);
\draw[ultra thick,color=blue] (2*0.935016,2*-0.35460) -- (2*0.663122,2*-0.74851);
\draw[ultra thick,color=blue] (2*0.663122,2*-0.74851) -- (2*0.239215,2*-0.970941);
\draw[ultra thick,color=blue] (2*0.239215,2*-0.970941) -- (2*-0.239215,2*-0.970941);
\draw[ultra thick,color=blue] (2*-0.239215,2*-0.970941) -- (2*-0.663122,2*-0.74851);
\draw[ultra thick,color=blue] (2*-0.663122,2*-0.74851) -- (2*-0.935016,2*-0.35460);
\draw[ultra thick,color=blue] (2*-0.935016,2*-0.35460) -- (2*-0.992708,2*0.12053);
\draw[ultra thick,color=blue] (2*-0.992708,2*0.12053) -- (2*-0.822983,2*0.56806);
\draw[ultra thick,color=blue] (2*-0.822983,2*0.56806) -- (2*-0.464723,2*0.88545);
\draw[ultra thick,color=blue] (2*-0.464723,2*0.88545) -- (0,2);
\draw (0,0.8) -- (0.8*0.951056,0.8*0.309016);
\draw (0.8*0.951056,0.8*0.309016) -- (0.8*0.587785,0.8*-0.809016);
\draw[ultra thick,dashed,color=red] (0.8*0.587785,0.8*-0.809016) -- (0.8*-0.587785,0.8*-0.809016);
\draw (0.8*-0.587785,0.8*-0.809016) -- (0.8*-0.951056,0.8*0.309016);
\draw (0.8*-0.951056,0.8*0.309016) -- (0,0.8);
\draw (0.8*0.587785,0.8*-0.809016) -- (2*0.663122,2*-0.74851);
\draw (0.8*-0.587785,0.8*-0.809016) -- (2*-0.935016,2*-0.35460);
\draw (0.8*0.951056,0.8*0.309016) -- (2*0.464723,2*0.88545);
\draw (0.8*-0.951056,0.8*0.309016) --  (2*-0.464723,2*0.88545);
\draw (0,2)-- (0,0.8);
\draw (2*0.992708,2*0.12053) --  (2*0.239215,2*-0.970941);
\draw (2*0.822983,2*0.56806) --  (2*0.935016,2*-0.35460);
\draw (2*-0.239215,2*-0.970941) --  (2*-0.992708,2*0.12053);
\draw (2*-0.663122,2*-0.74851) --   (2*-0.822983,2*0.56806);
\node at (0,2) [ circle,fill=black,draw=black,scale=0.4] {};
\node at (2*0.464723,2*0.88545) [ circle,fill=black,draw=black,scale=0.4] {};
\node at (2*0.822983,2*0.56806) [ circle,fill=black,draw=black,scale=0.4] {};
\node at (2*0.992708,2*0.12053) [ circle,fill=black,draw=black,scale=0.4] {};
\node at (2*0.935016,2*-0.35460) [ circle,fill=black,draw=black,scale=0.4] {};
\node at (2*0.663122,2*-0.74851) [ circle,fill=black,draw=black,scale=0.4] {};
\node at (2*0.239215,2*-0.970941) [ circle,fill=black,draw=black,scale=0.4] {};
\node at (2*-0.239215,2*-0.970941) [ circle,fill=black,draw=black,scale=0.4] {};
\node at (2*-0.663122,2*-0.74851) [ circle,fill=black,draw=black,scale=0.4] {};
\node at (2*-0.935016,2*-0.35460) [ circle,fill=black,draw=black,scale=0.4] {};
\node at (2*-0.992708,2*0.12053) [ circle,fill=black,draw=black,scale=0.4] {};
\node at (2*-0.822983,2*0.56806) [ circle,fill=black,draw=black,scale=0.4] {};
\node at (2*-0.464723,2*0.88545) [ circle,fill=black,draw=black,scale=0.4] {};

\node at (0,0.8) [ circle,fill=black,draw=black,scale=0.4] {};
\node at (0.8*0.951056,0.8*0.309016) [ circle,fill=black,draw=black,scale=0.4] {};
\node at (0.8*0.587785,0.8*-0.809016) [ circle,fill=black,draw=black,scale=0.4] {};
\node at (0.8*-0.587785,0.8*-0.809016) [ circle,fill=black,draw=black,scale=0.4] {};
\node at (0.8*-0.951056,0.8*0.309016) [ circle,fill=black,draw=black,scale=0.4] {};

\node at (0,-0.9) {$e$};

\end{tikzpicture}
\end{center}
\caption{An edge $e$ at distance at most $3$ of twelve edges from a 13-cycle (dashed line: $e$; thick lines: edges at distance at most $3$ from $e$ in the 13-cycle).}
\label{fig313}
\end{figure}
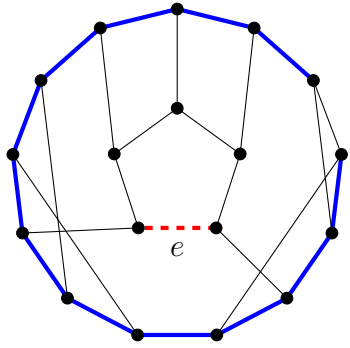

In the following proposition, we prove that cubic graphs of oddness $2$ are $(1,1,1,3)$-colorable in the case there are restrictions on the structure of the odd cycles of a $2$-factor.
Note that Property i) of the following proposition implies that every cubic graph of girth 13 and oddness $2$ is $(1,1,1,3)$-colorable.
\begin{prop}\label{1*3}
Every cubic graph of oddness $2$ having a $2$-factor containing two odd cycles $C_{1}$ and $C_{2}$ is $(1,1,1,3)$-colorable in the following cases:
\begin{itemize}
\item[i)] If $C_{1}$ or $C_{2}$ is a cycle of length at least $13$.
\item[ii)] If $\{C_{1},C_{2}\}$ contains a cycle $C$ of length at least $9$ and if there exists at least one edge with one extremity in $\{C_{1},C_{2}\}\setminus C$ and with the other extremity in $G-C$.
\item[iii)] If $\{C_{1},C_{2}\}$ contains a cycle $C$ of length at least $5$ and if there exists an edge of $\{C_{1},C_{2}\}\setminus C$ with both extremities having no neighbors in $C$.
\end{itemize}
\end{prop}
\begin{proof}
i) Suppose that $C_1$ has length $13$. Let $e$ be an edge from $C_2$. The edge $e$ is adjacent with two edges not belonging to $C_2$. Moreover, each of these two edges is possibly at distance at most $2$ of at most four edges of $C_1$. Also, there are possibly two edges at distance $2$ from $e$ which have one extremity in $C_2$ and the other extremity in $C_1$. Each of these two edges are adjacent with at most two edges of $C_1$. Thus, $e$ can be at distance at most $3$ from at most twelve edges of $C_1$. Figure \ref{fig313} illustrates an edge at distance at most $3$ from exactly twelve edges. Hence, we can color $e$ and one edge of $C_1$ at distance at least $4$ from $e$ with the color of radius $3$ and the remaining uncolored edges with the three colors of radius $1$.\newline
ii) Suppose that $C_1$ has length at least $9$ and there exists $xy\in E(G)$ with $x\in V(C_2)$ and $y\in V(G-C_1)$. Let $e$ be an edge of $C_2$ incident with $x$. In contrast with Property i), the edges adjacent with $e$ are at distance at most $2$ of at most four edges of $C_1$ (in Property i), it was eight). Thus, $e$ is possibly at distance at most $3$ from at most eight edges of $C_1$. Hence, we can color $e$ and one edge of $C_1$ at distance at least $4$ from $e$ with the color of radius $3$ and the remaining uncolored edges with the three colors of radius $1$. \newline
iii) Let $C_{1}$ be the cycle of length at least $5$ and let $e$ be an edge of $C_2$ with no extremity having a neighbor in $C_1$. In contrast with Property i), the edges adjacent with $e$ which are not in $C_{2}$ cannot be at distance at most $2$ of an edge of $C_1$ (in Property ii), $e$ was at distance at most $2$ of eight edges of $C_1$). Thus, $e$ is possibly at distance at most $3$ from at most four edges of $C_1$. Hence, we can color $e$ and one edge of $C_1$ at distance at least $4$ from $e$ with the color of radius $3$ and the remaining uncolored edges with the three colors of radius $1$.
\end{proof}

Since the flower snarks of order at least 20 have a 2-factor containing two cycles with one cycle among them of order at least 15, Proposition~\ref{1*3}.i) implies that flower snarks of order at least 20 are $(1,1,1,3)$-colorable.

Similarly, generalized Petersen graphs of order at least 26 have a $2$-factor containing two cycles with one cycle among them of order at least 13. Moreover, generalized Pertersen graphs of order which is a multiple of 4 are $3$-edge-colorable.
Thus, the generalized Petersen graphs of order at least 24 are $(1,1,1,3)$-colorable.

Note that the Pertersen and Tietze graphs are non $(1,1,1,3)$-colorable since they have oddness $2$ and the distance between any two edges in any $2$-factor of these graphs is at most $3$. 
By computer we have verified that every cubic graph of order at most 22, except the Petersen graph and the Tietze graph, is $(1,1,1,3)$-colorable. 

\begin{cor}
The Petersen graph is the only generalized Petersen graph which is not $(1,1,1,3)$-colorable and the Tietze graph is the only flower snark which is not $(1,1,1,3)$-colorable.
\end{cor}

Finally, in the next proposition, we prove that every cubic graph of oddness 4 is also $(1,1,1,3)$-colorable when the odd cycles of the 2-factor are sufficiently large.

\begin{prop}\label{11*3}
Every cubic graph of oddness $4$ having a $2$-factor with four odd cycles and with three cycles among them of length at least $13$ is $(1,1,1,3)$-colorable.
\end{prop}
\begin{proof}
Let $G$ be a cubic graph of oddness $4$ and let $\mathcal{F}'$ be the set containing the four odd cycles from the $2$-factor.
Let $C_{1}$ be the odd cycle having the minimum number of vertices in $\mathcal{F}'$.
As in proof of Theorem \ref{2*3}, we give labels to the vertices of $\mathcal{F}'$ as follows. If a vertex belonging to $C$ from $\mathcal{F}'$ has a neighbor in a different cycle of $\mathcal{F}'$, we label it by $+$, otherwise we label it by $-$. Figure \ref{figodd13} illustrates a labeling of the vertices of the odd cycles of a $2$-factor.

Let $e_{1}$ be an edge of $C_{1}$ such that there exists a cycle $C_{2}$ from $\mathcal{F}'$ such that all edges at distance at most $3$ from $e_{1}$ in $\mathcal{F}'$ are in $C_{1}\cup C_{2}$. Such an edge can be obtained by proceeding as in the proof of Theorem \ref{2*3}. As in the proof of Proposition \ref{1*3}.i), $e_{1}$ is at distance at most $3$ of at most twelve edges of $C_{2}$. Let $e_{2}$ be an edge of $C_{2}$ being at distance at least $4$ of $e_{1}$. Such an edge exists since $C_{2}$ has at least thirteen edges.

First, suppose there exists a cycle $C_{3}$ of $\mathcal{F}'$ such that all edges at distance at most $3$ from $e_{2}$ are in $C_{2}\cup C_{3}$. As previously, by Proposition \ref{1*3}.i), $e_{2}$ is at distance at most $3$ of at most twelve edges of of $C_{3}$. Let $e_{3}$ be an edge of $C_{3}$ being at distance at least $4$ of $e_{2}$. Also in this case, $e_{3}$ is at distance at most $3$ of at most twelve edges of $C_{4}$, $C_{4}$ being the odd cycle of $\mathcal{F}'-(C_{1}\cup C_{2}\cup C_{3})$. Consequently, there exists an edge $e_{4}$ of $C_{4}$ such that $e_{4}$ is at distance at least $4$ of $e_{3}$. Finally, $A=\{e_{1},e_{2},e_{3},e_{4}\}$ is a set of type I such that $G^{4}[A]$ is an empty graph.

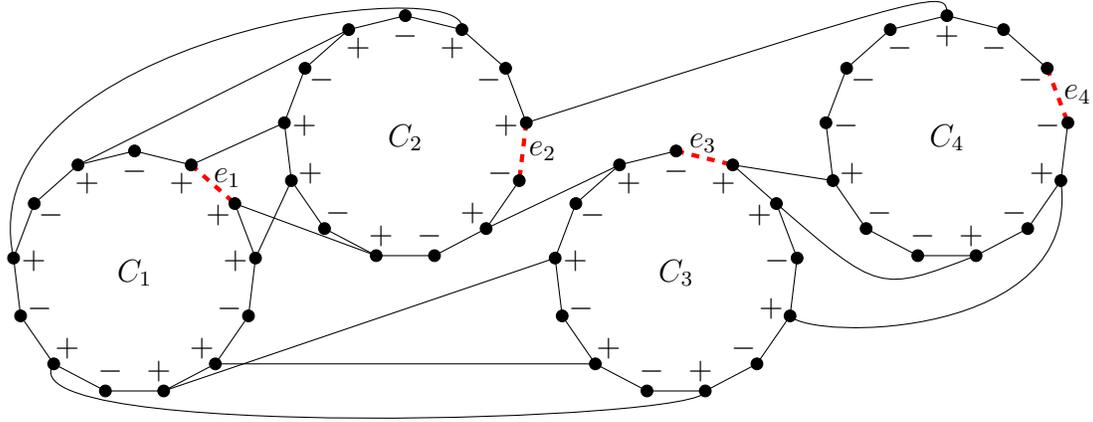
\begin{figure}[t]
\begin{center}
\begin{tikzpicture}[scale=0.9]
\draw (1.8*0.464723,1.8*0.88545) -- (1.8*-0.992708+4,1.8*0.12053+2);
\draw (1.8*0.822983,1.8*0.56806) -- (1.8*-0.239215+4,1.8*-0.970941+2);
\draw (1.8*0.992708,1.8*0.12053)  -- (1.8*-0.935016+4,1.8*-0.35460+2);
\draw (1.8*-0.464723,1.8*0.88545)  -- (1.8*-0.464723+4,1.8*0.88545+2);
\draw (1.8*0.663122+4,1.8*-0.74851+2) -- (1.8*-0.464723+8,1.8*0.88545);
\draw (1.8*0.663122,1.8*-0.74851)  -- (1.8*-0.663122+8,1.8*-0.74851);
\draw (1.8*0.239215,1.8*-0.970941)  -- (1.8*-0.992708+8,1.8*0.12053);
\draw (1.8*-0.992708,1.8*0.12053)  .. controls  (-2.5,3.5) and (5,4.5).. (1.8*0.464723+4,1.8*0.88545+2);
\draw (1.8*-0.663122,1.8*-0.74851)  .. controls  (-2,-2.5) and (8.5,-2.2) .. (1.8*0.239215+8,1.8*-0.970941);
\draw (1.8*0.992708+4,1.8*0.12053+2) .. controls  (12,4.2) .. (12,1.8+2);
\draw (1.8*0.464723+8,1.8*0.88545) -- (1.8*-0.935016+12,1.8*-0.35460+2);
\draw (1.8*0.822983+8,1.8*0.56806) .. controls (11,-0.3).. (1.8*0.239215+12,1.8*-0.970941+2);
\draw  (1.8*0.935016+8,1.8*-0.35460) .. controls (10,-1) and (14,-1).. (1.8*0.935016+12,1.8*-0.35460+2);

\draw (0,1.8) -- (1.8*0.464723,1.8*0.88545);
\draw[ultra thick,dashed,color=red] (1.8*0.464723,1.8*0.88545) -- (1.8*0.822983,1.8*0.56806);
\draw (1.8*0.822983,1.8*0.56806) -- (1.8*0.992708,1.8*0.12053);
\draw (1.8*0.992708,1.8*0.12053) -- (1.8*0.935016,1.8*-0.35460);
\draw (1.8*0.935016,1.8*-0.35460) -- (1.8*0.663122,1.8*-0.74851);
\draw (1.8*0.663122,1.8*-0.74851) -- (1.8*0.239215,1.8*-0.970941);
\draw (1.8*0.239215,1.8*-0.970941) -- (1.8*-0.239215,1.8*-0.970941);
\draw (1.8*-0.239215,1.8*-0.970941) -- (1.8*-0.663122,1.8*-0.74851);
\draw (1.8*-0.663122,1.8*-0.74851) -- (1.8*-0.935016,1.8*-0.35460);
\draw (1.8*-0.935016,1.8*-0.35460) -- (1.8*-0.992708,1.8*0.12053);
\draw (1.8*-0.992708,1.8*0.12053) -- (1.8*-0.822983,1.8*0.56806);
\draw (1.8*-0.822983,1.8*0.56806) -- (1.8*-0.464723,1.8*0.88545);
\draw (1.8*-0.464723,1.8*0.88545) -- (0,1.8);

\node at (0,1.8) [ circle,fill=black,draw=black,scale=0.4] {};
\node at (1.8*0.464723,1.8*0.88545) [ circle,fill=black,draw=black,scale=0.4] {};
\node at (1.8*0.822983,1.8*0.56806) [ circle,fill=black,draw=black,scale=0.4] {};
\node at (1.8*0.992708,1.8*0.12053) [ circle,fill=black,draw=black,scale=0.4] {};
\node at (1.8*0.935016,1.8*-0.35460) [ circle,fill=black,draw=black,scale=0.4] {};
\node at (1.8*0.663122,1.8*-0.74851) [ circle,fill=black,draw=black,scale=0.4] {};
\node at (1.8*0.239215,1.8*-0.970941) [ circle,fill=black,draw=black,scale=0.4] {};
\node at (1.8*-0.239215,1.8*-0.970941) [ circle,fill=black,draw=black,scale=0.4] {};
\node at (1.8*-0.663122,1.8*-0.74851) [ circle,fill=black,draw=black,scale=0.4] {};
\node at (1.8*-0.935016,1.8*-0.35460) [ circle,fill=black,draw=black,scale=0.4] {};
\node at (1.8*-0.992708,1.8*0.12053) [ circle,fill=black,draw=black,scale=0.4] {};
\node at (1.8*-0.822983,1.8*0.56806) [ circle,fill=black,draw=black,scale=0.4] {};
\node at (1.8*-0.464723,1.8*0.88545) [ circle,fill=black,draw=black,scale=0.4] {};

\draw[ultra thick,dashed,color=red] (0+8,1.8) -- (1.8*0.464723+8,1.8*0.88545);
\draw (1.8*0.464723+8,1.8*0.88545) -- (1.8*0.822983+8,1.8*0.56806);
\draw (1.8*0.822983+8,1.8*0.56806) -- (1.8*0.992708+8,1.8*0.12053);
\draw (1.8*0.992708+8,1.8*0.12053) -- (1.8*0.935016+8,1.8*-0.35460);
\draw (1.8*0.935016+8,1.8*-0.35460) -- (1.8*0.663122+8,1.8*-0.74851);
\draw (1.8*0.663122+8,1.8*-0.74851) -- (1.8*0.239215+8,1.8*-0.970941);
\draw (1.8*0.239215+8,1.8*-0.970941) -- (1.8*-0.239215+8,1.8*-0.970941);
\draw (1.8*-0.239215+8,1.8*-0.970941) -- (1.8*-0.663122+8,1.8*-0.74851);
\draw (1.8*-0.663122+8,1.8*-0.74851) -- (1.8*-0.935016+8,1.8*-0.35460);
\draw (1.8*-0.935016+8,1.8*-0.35460) -- (1.8*-0.992708+8,1.8*0.12053);
\draw (1.8*-0.992708+8,1.8*0.12053) -- (1.8*-0.822983+8,1.8*0.56806);
\draw (1.8*-0.822983+8,1.8*0.56806) -- (1.8*-0.464723+8,1.8*0.88545);
\draw (1.8*-0.464723+8,1.8*0.88545) -- (0+8,1.8);

\node at (0+8,1.8) [ circle,fill=black,draw=black,scale=0.4] {};
\node at (1.8*0.464723+8,1.8*0.88545) [ circle,fill=black,draw=black,scale=0.4] {};
\node at (1.8*0.822983+8,1.8*0.56806) [ circle,fill=black,draw=black,scale=0.4] {};
\node at (1.8*0.992708+8,1.8*0.12053) [ circle,fill=black,draw=black,scale=0.4] {};
\node at (1.8*0.935016+8,1.8*-0.35460) [ circle,fill=black,draw=black,scale=0.4] {};
\node at (1.8*0.663122+8,1.8*-0.74851) [ circle,fill=black,draw=black,scale=0.4] {};
\node at (1.8*0.239215+8,1.8*-0.970941) [ circle,fill=black,draw=black,scale=0.4] {};
\node at (1.8*-0.239215+8,1.8*-0.970941) [ circle,fill=black,draw=black,scale=0.4] {};
\node at (1.8*-0.663122+8,1.8*-0.74851) [ circle,fill=black,draw=black,scale=0.4] {};
\node at (1.8*-0.935016+8,1.8*-0.35460) [ circle,fill=black,draw=black,scale=0.4] {};
\node at (1.8*-0.992708+8,1.8*0.12053) [ circle,fill=black,draw=black,scale=0.4] {};
\node at (1.8*-0.822983+8,1.8*0.56806) [ circle,fill=black,draw=black,scale=0.4] {};
\node at (1.8*-0.464723+8,1.8*0.88545) [ circle,fill=black,draw=black,scale=0.4] {};

\draw (0+4,1.8+2) -- (1.8*0.464723+4,1.8*0.88545+2);
\draw (1.8*0.464723+4,1.8*0.88545+2) -- (1.8*0.822983+4,1.8*0.56806+2);
\draw (1.8*0.822983+4,1.8*0.56806+2) -- (1.8*0.992708+4,1.8*0.12053+2);
\draw[ultra thick,dashed,color=red] (1.8*0.992708+4,1.8*0.12053+2) -- (1.8*0.935016+4,1.8*-0.35460+2);
\draw (1.8*0.935016+4,1.8*-0.35460+2) -- (1.8*0.663122+4,1.8*-0.74851+2);
\draw (1.8*0.663122+4,1.8*-0.74851+2) -- (1.8*0.239215+4,1.8*-0.970941+2);
\draw (1.8*0.239215+4,1.8*-0.970941+2) -- (1.8*-0.239215+4,1.8*-0.970941+2);
\draw (1.8*-0.239215+4,1.8*-0.970941+2) -- (1.8*-0.663122+4,1.8*-0.74851+2);
\draw (1.8*-0.663122+4,1.8*-0.74851+2) -- (1.8*-0.935016+4,1.8*-0.35460+2);
\draw (1.8*-0.935016+4,1.8*-0.35460+2) -- (1.8*-0.992708+4,1.8*0.12053+2);
\draw (1.8*-0.992708+4,1.8*0.12053+2) -- (1.8*-0.822983+4,1.8*0.56806+2);
\draw (1.8*-0.822983+4,1.8*0.56806+2) -- (1.8*-0.464723+4,1.8*0.88545+2);
\draw (1.8*-0.464723+4,1.8*0.88545+2) -- (0+4,1.8+2);

\node at (0+4,1.8+2) [ circle,fill=black,draw=black,scale=0.4] {};
\node at (1.8*0.464723+4,1.8*0.88545+2) [ circle,fill=black,draw=black,scale=0.4] {};
\node at (1.8*0.822983+4,1.8*0.56806+2) [ circle,fill=black,draw=black,scale=0.4] {};
\node at (1.8*0.992708+4,1.8*0.12053+2) [ circle,fill=black,draw=black,scale=0.4] {};
\node at (1.8*0.935016+4,1.8*-0.35460+2) [ circle,fill=black,draw=black,scale=0.4] {};
\node at (1.8*0.663122+4,1.8*-0.74851+2) [ circle,fill=black,draw=black,scale=0.4] {};
\node at (1.8*0.239215+4,1.8*-0.970941+2) [ circle,fill=black,draw=black,scale=0.4] {};
\node at (1.8*-0.239215+4,1.8*-0.970941+2) [ circle,fill=black,draw=black,scale=0.4] {};
\node at (1.8*-0.663122+4,1.8*-0.74851+2) [ circle,fill=black,draw=black,scale=0.4] {};
\node at (1.8*-0.935016+4,1.8*-0.35460+2) [ circle,fill=black,draw=black,scale=0.4] {};
\node at (1.8*-0.992708+4,1.8*0.12053+2) [ circle,fill=black,draw=black,scale=0.4] {};
\node at (1.8*-0.822983+4,1.8*0.56806+2) [ circle,fill=black,draw=black,scale=0.4] {};
\node at (1.8*-0.464723+4,1.8*0.88545+2) [ circle,fill=black,draw=black,scale=0.4] {};

\draw (0+12,1.8+2) -- (1.8*0.464723+12,1.8*0.88545+2);
\draw (1.8*0.464723+12,1.8*0.88545+2) -- (1.8*0.822983+12,1.8*0.56806+2);
\draw[ultra thick,dashed,color=red] (1.8*0.822983+12,1.8*0.56806+2) -- (1.8*0.992708+12,1.8*0.12053+2);
\draw (1.8*0.992708+12,1.8*0.12053+2) -- (1.8*0.935016+12,1.8*-0.35460+2);
\draw (1.8*0.935016+12,1.8*-0.35460+2) -- (1.8*0.663122+12,1.8*-0.74851+2);
\draw (1.8*0.663122+12,1.8*-0.74851+2) -- (1.8*0.239215+12,1.8*-0.970941+2);
\draw (1.8*0.239215+12,1.8*-0.970941+2) -- (1.8*-0.239215+12,1.8*-0.970941+2);
\draw (1.8*-0.239215+12,1.8*-0.970941+2) -- (1.8*-0.663122+12,1.8*-0.74851+2);
\draw (1.8*-0.663122+12,1.8*-0.74851+2) -- (1.8*-0.935016+12,1.8*-0.35460+2);
\draw (1.8*-0.935016+12,1.8*-0.35460+2) -- (1.8*-0.992708+12,1.8*0.12053+2);
\draw (1.8*-0.992708+12,1.8*0.12053+2) -- (1.8*-0.822983+12,1.8*0.56806+2);
\draw (1.8*-0.822983+12,1.8*0.56806+2) -- (1.8*-0.464723+12,1.8*0.88545+2);
\draw (1.8*-0.464723+12,1.8*0.88545+2) -- (0+12,1.8+2);

\node at (0+12,1.8+2) [ circle,fill=black,draw=black,scale=0.4] {};
\node at (1.8*0.464723+12,1.8*0.88545+2) [ circle,fill=black,draw=black,scale=0.4] {};
\node at (1.8*0.822983+12,1.8*0.56806+2) [ circle,fill=black,draw=black,scale=0.4] {};
\node at (1.8*0.992708+12,1.8*0.12053+2) [ circle,fill=black,draw=black,scale=0.4] {};
\node at (1.8*0.935016+12,1.8*-0.35460+2) [ circle,fill=black,draw=black,scale=0.4] {};
\node at (1.8*0.663122+12,1.8*-0.74851+2) [ circle,fill=black,draw=black,scale=0.4] {};
\node at (1.8*0.239215+12,1.8*-0.970941+2) [ circle,fill=black,draw=black,scale=0.4] {};
\node at (1.8*-0.239215+12,1.8*-0.970941+2) [ circle,fill=black,draw=black,scale=0.4] {};
\node at (1.8*-0.663122+12,1.8*-0.74851+2) [ circle,fill=black,draw=black,scale=0.4] {};
\node at (1.8*-0.935016+12,1.8*-0.35460+2) [ circle,fill=black,draw=black,scale=0.4] {};
\node at (1.8*-0.992708+12,1.8*0.12053+2) [ circle,fill=black,draw=black,scale=0.4] {};
\node at (1.8*-0.822983+12,1.8*0.56806+2) [ circle,fill=black,draw=black,scale=0.4] {};
\node at (1.8*-0.464723+12,1.8*0.88545+2) [ circle,fill=black,draw=black,scale=0.4] {};

\node at (0,1.5)  {$-$};
\node at (1.5*0.464723,1.5*0.88545)  {$+$};
\node at (2.2*0.622983,2*0.70)  {$e_1$};
\node at (1.5*0.822983,1.5*0.56806)  {$+$};
\node at (1.5*0.992708,1.5*0.12053)  {$+$};
\node at (1.5*0.935016,1.5*-0.35460) {$-$};
\node at (1.5*0.663122,1.5*-0.74851)  {$+$};
\node at (1.5*0.239215,1.5*-0.970941) {$+$};
\node at (1.5*-0.239215,1.5*-0.970941) {$-$};
\node at (1.5*-0.663122,1.5*-0.74851) {$+$};
\node at (1.5*-0.935016,1.5*-0.35460) {$-$};
\node at (1.5*-0.992708,1.5*0.12053) {$+$};
\node at (1.5*-0.822983,1.5*0.56806) {$-$};
\node at (1.5*-0.464723,1.5*0.88545) {$+$};

\node at (0+4,1.5+2)  {$-$};
\node at (1.5*0.464723+4,1.5*0.88545+2)  {$+$};
\node at (1.5*0.822983+4,1.5*0.56806+2)  {$-$};
\node at (1.5*0.992708+4,1.5*0.12053+2)  {$+$};
\node at (2.1*0.965016+4,2*-0.10+2)  {$e_2$};
\node at (1.5*0.935016+4,1.5*-0.35460+2) {$-$};
\node at (1.5*0.663122+4,1.5*-0.74851+2)  {$+$};
\node at (1.5*0.239215+4,1.5*-0.970941+2) {$-$};
\node at (1.5*-0.239215+4,1.5*-0.970941+2) {$+$};
\node at (1.5*-0.663122+4,1.5*-0.74851+2) {$-$};
\node at (1.5*-0.935016+4,1.5*-0.35460+2) {$+$};
\node at (1.5*-0.992708+4,1.5*0.12053+2) {$+$};
\node at (1.5*-0.822983+4,1.5*0.56806+2) {$-$};
\node at (1.5*-0.464723+4,1.5*0.88545+2) {$+$};

\node at (0+8,1.5)  {$-$};
\node at (1.6*0.494723/2+8,1.6*1.2)  {$e_3$};
\node at (1.5*0.464723+8,1.5*0.88545)  {$+$};
\node at (1.5*0.822983+8,1.5*0.56806)  {$+$};
\node at (1.5*0.992708+8,1.5*0.12053)  {$-$};
\node at (1.5*0.935016+8,1.5*-0.35460) {$+$};
\node at (1.5*0.663122+8,1.5*-0.74851)  {$-$};
\node at (1.5*0.239215+8,1.5*-0.970941) {$+$};
\node at (1.5*-0.239215+8,1.5*-0.970941) {$-$};
\node at (1.5*-0.663122+8,1.5*-0.74851) {$+$};
\node at (1.5*-0.935016+8,1.5*-0.35460) {$-$};
\node at (1.5*-0.992708+8,1.5*0.12053) {$+$};
\node at (1.5*-0.822983+8,1.5*0.56806) {$-$};
\node at (1.5*-0.464723+8,1.5*0.88545) {$+$};

\node at (0+12,1.5+2)  {$+$};
\node at (1.5*0.464723+12,1.5*0.88545+2)  {$-$};
\node at (1.5*0.822983+12,1.5*0.56806+2)  {$-$};
\node at (2*0.965016+12,2*0.32+2)  {$e_4$};
\node at (1.5*0.992708+12,1.5*0.12053+2)  {$-$};
\node at (1.5*0.935016+12,1.5*-0.35460+2) {$+$};
\node at (1.5*0.663122+12,1.5*-0.74851+2)  {$-$};
\node at (1.5*0.239215+12,1.5*-0.970941+2) {$+$};
\node at (1.5*-0.239215+12,1.5*-0.970941+2) {$-$};
\node at (1.5*-0.663122+12,1.5*-0.74851+2) {$-$};
\node at (1.5*-0.935016+12,1.5*-0.35460+2) {$+$};
\node at (1.5*-0.992708+12,1.5*0.12053+2) {$-$};
\node at (1.5*-0.822983+12,1.5*0.56806+2) {$-$};
\node at (1.5*-0.464723+12,1.5*0.88545+2) {$-$};

\node at (0,0) {$C_1$};
\node at (4,2) {$C_2$};
\node at (8,0) {$C_3$};
\node at (12,2) {$C_4$};
\end{tikzpicture}
\end{center}
\caption{An illustration of the odd cycles $\{C_{1},C_{2},C_{3},C_{4}\}$ from a 2-factor considered in proof of Proposition \ref{11*3} (dashed lines: edges of $A$; the edge having end vertices labeled by $-$ in $G-\mathcal{F}$ are omitted in order to simplify the figure).}
\label{figodd13}
\end{figure}

Second, by excluding the first case, the two odd cycles $C_{3}$ and $C_4$ of $\mathcal{F}'-(C_{1}\cup C_{2})$ both contain edges at distance at most $3$ from $e_{2}$. Figure~\ref{figodd13} illustrates the edge $e_{2}$ in this case.
Let $u_1$ and $u_{2}$ be the extremities of $e_{2}$ and suppose $u_0$, $u_1$, $u_2$ and $u_{3}$ are consecutive vertices of $C_{2}$. By hypothesis and by Proposition \ref{fouqq}, these vertices are either labeled by $+,-,+,+$, or by $+,-,+,-$, or by $-,+,-,+$ or by $+,+,-,+$. Up to symmetry, we can suppose that these vertices are either labeled by $+,-,+,+$ or by $+,-,+,-$ (as in Figure~\ref{figodd13}). Let $C_{3}$ be the odd cycle of $\mathcal{F}'-(C_{1}\cup C_{2})$ containing a neighbor of $u_{0}$. By excluding the first case, $u_2$ and $u_3$ have no neighbor in $C_3$. Let $v_{0}$ be the neighbor of $u_0$ in $C_{3}$. Suppose $v_{0}$, $v_{1}$, $v_{2}$ and $v_{3}$ are consecutive vertices of $C_{3}$ and let $e_{3}=v_{1}v_{2}$. By Hypothesis $v_{0}$ is labeled by $+$. Note that  $v_{0}$, $v_{1}$, $v_{2}$ and $v_{3}$ are labeled consecutively by $+,-,+,+$ (as in Figure~\ref{figodd13}), or by $+,-,-,+$, or by $+,-,-,+$ or by $+,-,-,-$.
Consequently, there are at most six edges at distance at most $3$ from $e_{2}$ in $C_{4}$ and at most six edges at distance at most $3$ from $e_{3}$ in $C_{4}$. Since $C_{4}$ has length at least thirteen, there remains an edge $e_{4}$ in $C_4$ such that $e_{4}$ is at distance at least $4$ from both $e_{2}$ and $e_{3}$, as Figure~\ref{figodd13} illustrates. Finally, $A=\{e_{1},e_{2},e_{3},e_{4}\}$ is a set of type I such that $G^{4}[A]$ is an empty graph and Lemma~\ref{lem1111}.i) allows to conclude.
\end{proof}

We have performed computations on cubic graphs with small order that lead us to believe that all cubic graphs, except two snarks, are $(1,1,1,3)$-colorable. 
\begin{con}\label{con2}
Every cubic graph different from the Petersen graph and from the Tietze graph is $(1,1,1,3)$-colorable.
\end{con}
A weaker form of the previous conjecture is to restrict to cubic graphs of oddness 2 or 4.

\subsection{$(1,1,1,4,\ldots,4)$-coloring}
As in the previous subsection, we try in this subsection to minimize the number of required integers 4 in order that all cubic graphs having a $2$-factor are $(1,1,1,4,\ldots,4)$-colorable. We give a first result about this problem, however it is probably not tight since we have not been able to find a non $(1,1,1,4,4)$-colorable cubic graph. 
\begin{theo}\label{5*4}
Every cubic graph having a $2$-factor is $(1,1,1,4^{5})$-colorable.
\end{theo}

\begin{proof}

Let $G$ be a cubic graph. Let $\mathcal{F}$ be $2$-factor of $G$ having a minimum number of odd cycles. Let $\mathcal{F}'$ be the set of odd cycles from $\mathcal{F}$. By Lemma~\ref{lem1111}.i) and by Brooks' theorem, it is sufficient to exhibit a set $A$ of type I in $\mathcal{F}$ such that $\Delta(G^{4}[A])\le 4$ in order $G$ to be $(1,1,1,4^{5})$-colorable.

As in the proof of Theorem \ref{2*3}, we give labels to the vertices of $\mathcal{F}'$ as follows. If a vertex belonging to a cycle $C$ from $\mathcal{F}'$ has a neighbor in a different cycle of $\mathcal{F}'$, we label it by $+$, otherwise we label it by $-$. 
By Proposition \ref{fouqq}, the two end vertices of any edge of any cycle $C$ of $\mathcal{F'}$ have neighbors in only at most one cycle of $\mathcal{F'}$ other than $C$. Thus, if there are consecutive vertices labeled by $+$ in an odd cycle $C$ from $\mathcal{F}'$, then there exists an unique cycle $C'$ of $\mathcal{F}'$ such that all these vertices have neighbors only in $C\cup C'$.

Let $u$ be a vertex of a odd cycle $C$ from $\mathcal{F}'$ and suppose that $u_{0}$, $u$, $u_{1}$ are consecutive vertices of $C$. Let $e_{0}$ be the edge incident with $u_{0}$ in $G-\mathcal{F}'$ and let $e_{1}$ be the edge incident with $u_{1}$ in $G-\mathcal{F}'$.
We make the following observation in which $D^{C}_{3}(u)$ denotes the set of odd cycles from $\mathcal{F}'-C$ containing a vertex at distance at most $3$ from $u$ in the subgraph $G-\{e_{0}, e_{1}\}$.

\vspace{0.5cm}
\textbf{Observation 2.} For any vertex $u$ in an odd cycle $C$ from $\mathcal{F}'$ the following is true:
\begin{enumerate}
\item[i)] if $u$ is labeled by $+$, then $|D^{C}_{3}(u)|=1$;
\item[ii)] if $u$ is labeled by $-$, then $|D^{C}_{3}(u)|\le 2$.
\end{enumerate}

Observation 2 can be easily obtained by observing that in the case $u$ is labeled by $+$, the other extremity of the edge incident with $u$ in $G-\mathcal{F}'$ is labeled by $+$ and its neighbors are either labeled by $+$ (and, in this case, are adjacent with vertices of $C$) or by $-$.

We will construct $A$ as follows. Since each cycle of $\mathcal{F}'$ has an odd number of vertices there exists two consecutive vertices labeled by $+$ or by $-$ in every cycle of $\mathcal{F}'$. For each cycle $C$ of $\mathcal{F}'$, we add an edge of $C$ into $A$ depending on the existence of two consecutive vertices labeled by $+$ and the existence of vertices labeled by $+$.

If there exist two consecutive vertices $u_{0}$ and $u_{1}$ which are both labeled by $+$ in $C$, then we add the edge $u_{0}u_{1}$ into $A$. Let $v_{0}$ be a vertex adjacent to $u_{0}$ and let $v_{1}$ be a vertex adjacent to $u_{1}$, $v_{0}$ and $v_{1}$ being in a cycle $C'$ from $\mathcal{F}'-C$. 
Let $e_{0}=u_{0}v_{0}$ and let $e_1=u_{1}v_{1}$. By Observation 2, there are no edges at distance less than $4$ of either $e_{0}$ or $e_{1}$ in $\mathcal{F'}-(C\cup C')$. Finally, suppose $u_{-2}$, $u_{-1}$, $u_{0}$, $u_{1}$, $u_{2}$, $u_{3}$ are consecutive vertices of $C$.
The vertices $u_{-2}$, $u_{-1}$, $u_{0}$, $u_{1}$, $u_{2}$ and $u_{3}$ are labeled in the worst case by $+,-,+,+,-,+$. Effectively, if $u_{-1}$ or $u_{2}$ is labeled by $+$, then, by Proposition~\ref{fouqq}, its neighbors should be in $C\cup C'$.
Thus, in the other cases, there are less vertices among  $u_{-2}$, $u_{-1}$, $u_{2}$ and $u_{3}$ which are labeled by $+$ and without having a neighbor in $C'$.
Therefore, $u_{0}u_{1}$ has degree at most $3$ in $G^{4}[A]$.

Suppose now that every vertex of $C$ is labeled by $-$. Let $e=uv$ be an edge of $C$. As in proof of Theorem \ref{2*3}, no edge of $C$ can be at distance less than 4 of an edge of $\mathcal{F}'-C$. By Observation 2 on the two vertices $u$ and $v$, we obtain that $e$ has degree at most $4$ in $G^{4}[A]$.

By excluding the two previous cases and since $C$ has odd length, there exists four consecutive vertices $v_{-2}$, $v_{-1}$, $v_{0}$, $v_1$ labeled by $-$, $-$, $+$, $-$. Let $v_{2}$ and $v_3$ be the vertices such that $v_{-2}$, $\ldots$, $v_{2}$, $v_{3}$ are consecutive in $C$. In this case, we add $v_{0} v_{1}$ into $A$.
Without loss of generality, we can suppose that, by excluding the first case, there are no two consecutive vertices of $C$ labeled by $+$.
 Thus, $v_{-2}$, $v_{-1}$, $v_{0}$, $v_{1}$, $v_{2}$ and $v_{3}$ are labeled consecutively by $-,-,+,-,+,-$, or by $-,-,+,-,-,+$ or by $-,-,+,-,-,-$. Thus, by Observation 2, $v_{0}v_{1}$ has degree at most 4 in $G^{4}[A]$. 
\end{proof}
We have performed computations on cubic graphs with small order that lead us to believe that all cubic graphs are $(1,1,1,4,4)$-colorable. 
\begin{con}\label{con2}
Every cubic graph is $(1,1,1,4,4)$-colorable.
\end{con}
Trivially, this conjecture holds for cubic graphs of oddness $2$ and a weaker form of the previous conjecture is to restrict the conjecture to cubic graphs of oddness 4.
Note that we have found, performing computations, four bridgeless cubic graphs of order 14 which are $(1,1,1,3)$-colorable but are not $(1,1,1,4)$-colorable.

Also, we state the following open problem which has already been solved for $k\le 2$~\cite{Fou2}.

\begin{quest}
For every integer $k\ge3$, is it true that there exists a cubic graph of arbitrary large order which is not $(1,1,1,k)$-colorable ?
\end{quest}

\section{$(1,1,k,\ldots,k)$-coloring}
As for $(1,1,1,k,\ldots,k)$-colorings, we first show a general result for arbitrary $k$. For this, we consider the following two sequences of integers:
let $(b_k)_{k\ge 2}$ and $(c_k)_{k\ge 1}$ be sequences of integers defined for $k\ge 2$ by $c_1=0$, $c_k=  2^k- c_{k-1}$ and $b_k=\sum_{i=1}^{k} c_i$. The sequence $(b_k)$ is known in OEIS under the number A153772 and has a closed-form formul\ae: $b_k=\frac{2^{k+2}+2(-1)^{k+2}-6}{3}$.
\begin{lem}\label{lem2}
Let $G$ be a cubic graph having a $2$-factor $\mathcal{F}$ and let $B\subseteq E(\mathcal{F})$ be of type II.
For $k\ge 2$, the graph $G^{k}[B]$ satisfies $\Delta (G^{k}[B])\le b_k$.
\end{lem}
\begin{proof}

Let $e\in B$.
The proof consists in showing that there are at most $b_k$ edges at distance at most $k$ from $e$ in $G$. The rest of the proof is divided into three claims.
\begin{description}

\item[Claim 1.] There are at most $2^{k+1}$ edges of $G$ at distance $k$ from $e$.

There are four edges adjacent to $e$ since $G$ is cubic. Suppose now that there are at most $2^{k}$ edges at distance $k$ from $e$. Since each edge at distance $k$ from $e$ is adjacent to at most two edges which are not at distance at most $k$ from $e$, we obtain that there are at most $2^{k+1}$ edges at distance $k+1$ from $e$.

\item[Claim 2.] If any $3$-edge-colorable cubic graph $G'$ with any $2$-factor $\mathcal{F}'$ of $G'$ only containing even cycles, and any set $B'\subseteq E(\mathcal{F}')$ of type II is such that $\Delta (G'^{k}[B'])\le b_k$, then $\Delta (G^{k}[B])\le b_k$.

Remark that each odd cycle of $\mathcal{F}$ contains two adjacent edges of $E(G)\setminus B$.
We create a new graph $G'$ by replacing, in each odd cycle of $\mathcal{F}$, one of these two adjacent edges of $E(G)\setminus B$ by a path of length $2$, i.e., by two adjacent edges. The vertices of degree $2$ created are connected together by a perfect matching (the number of created vertices of degree $2$ is even since the number of odd cycles in $\mathcal{F}$  is even).
We construct $B'$ from $B$ by adding, for each odd cycle of  $\mathcal{F}$, an edge of this cycle which has no neighboring edge of $B$ to $B'$.

We easily remark that $\Delta (G'^{k}[B'])\le b_k$  implies $\Delta (G^{k}[B])\le b_k$.

\item[Claim 3.] Suppose $G$ is a $3$-edge colorable graph and $\mathcal{F}$ only contains even cycles. There are at most $c_k$ edges at distance $k$ from $e$ in $B$ and at most $2^{k+1}-c_k$ edges at distance $k$ from $e$ in $E(G)\setminus B$, for $c_1=0$ and $c_k= 2^k- c_{k-1}$, $k\ge 2$.

Let $n_A^k$ and $n_{E\setminus A}^k$ be the number of edges at distance $k$ from $e$ in $B$ and in $E(G)\setminus B$, respectively.
There are four edges adjacent to $e$, each of these edges not being in $B$. Thus $n_A^1=0$ and $n_{E\setminus A}^1=4$.
Now suppose by induction that $n_A^k \le c^k$ and $n_{E\setminus A}^k\le 2^{k+1}-c_k$.
Observe that any edge from $e'\in B$ has no neighbor in $B$ and at most two neighboring edges in $E(G)\setminus B$ which are at greater distance from $e$ than $e'$. Also, observe that an edge $e''\in E(G)\setminus B$ has at most one neighbor in $B$ which is at greater distance from $e$ than $e''$ and at most one neighbor in $E(G)\setminus B$ which is at greater distance from $e$ than $e''$.
Thus we obtain that $n_A^{k+1}\le n_{E\setminus A}^k\le 2^{k+1}-c_k=c_ {k+1}$ and that $n_{E\setminus A}^{k+1}\le 2 n_A^k+ n_{E\setminus A}^k\le 2 c_k +2^{k+1}-c_k = c_k +2^{k+1}= 2^{k+2}- 2^{k+1} +c_k=2^{k+2} - c_{k+1}$.
\end{description}

These three claims together allow to obtain that there are at most $b_k$ edges of $A$ at distance at most $k$ from $e$.
Therefore, the graph $G^{k}[B]$ satisfies $\Delta (G^{k}[B])\le b_k$.
\end{proof}

Combining this lemma with the ones of Section 2, we obtain the following theorem.
\begin{theo}\label{akbk}
For any $k\ge 2$, every cubic graph having a $2$-factor is $(1,1,k^{a_{k}+b_k+2})$-colorable and every  $3$-edge-colorable cubic graph having a $2$-factor is $(1,1,k^{b_k+1})$-colorable, where $a_k=\frac{2^{k+1}-(-1)^{k+1}-3}{3}$ and $b_k=\frac{2^{k+2}+2(-1)^{k+2}-6}{3}$.
\end{theo}
\begin{proof}
Let $G$ be a cubic graph and $k\ge 2$ be an integer. Let $\mathcal{F}$ be a $2$-factor of $G$.
Let $A\subseteq E(\mathcal{F})$ be of type I. By Lemma \ref{lem1}, $\Delta (G^{k}[A])\le a_k$. Thus, by Brooks' theorem, $G^{k}[A]$ is $(a_{k}+1)$-colorable.
Let $B\subseteq E(\mathcal{F})\setminus A$ be of type II. By Lemma \ref{lem2}, $\Delta (G^{k}[B])\le b_k$. Thus, by Brooks' theorem, $G^{k}[B]$ is $(b_{k}+1)$-colorable.

Therefore, by Lemma~\ref{lem1111}.ii), $G$ is $(1,1,k^{b_{k}+a_k+2})$-colorable and even $(1,1,k^{b_{k}+1})$-colorable in the case $G$ is $3$-edge-colorable (since then, $\mathcal{F}$ can be chosen such that $A$ is empty).
\end{proof}
\subsection{$(1,1,2,\ldots,2)$-coloring}

The following proposition allows to refine our results for sequences of the form $(1,1,2,\ldots,2)$.
\begin{prop}\label{K5}
Let $G$ be a cubic graph of order at least $12$. If there exists a $2$-factor $\mathcal{F}$ in $G$, then for any set $B\subseteq E(\mathcal{F})$ of type II, $G^{2}[B]$ contains no connected component isomorphic to $K_5$.
\end{prop}
\begin{proof}
We first prove that if five edges of $B$ form a connected component isomorphic to $K_5$ in $G^{2}[B]$ then these edges lie on the same cycle of $\mathcal{F}$.
Suppose to the contrary, that $B_5$ is a set containing five edges of $B$ forming a $K_5$ in $G^{2}[B]$ and that there are, without loss of generality, two cycles $C_1$ and $C_2$ of $\mathcal{F}$ such that at least three edges of $B_5$ are in $C_1$ and at least one edge of $B_5$ is in $C_2$.
An edge $e=xy$ of $C_2$ cannot be at distance 2 of more than two edges of $B_5\cap C_1$ (the edges of $C_1$ at distance 2 from $e$ are those who have one endpoint adjacent to $x$ or to $y$). Hence, $e$ is not adjacent (in $G^{2}[B]$) to at least one edge of $B_5\cap C_1$, a contradiction.

Now, we show that if $G^{2}[B]$ contains a $K_5$, then $G$ has order 10, which contradicts the hypothesis. Let $C$ be the cycle containing the five edges of $B$ forming a $K_5$.
Note that if an edge of $B\cap C$ has one extremity which is adjacent to a vertex of $G\setminus C$, then this edge is at distance (in the cycle) at most $2$ of at most three other edges from $B\cap C$. Therefore, each vertex of $C$ has its neighbors in $C$ and, consequently, $C$ must have length 10. Since the graph $G$ is connected, $G$ has order 10.
\end{proof}

In any cubic graph $G$ having a $2$-factor, we can find a $2$-factor containing a minimum number of odd cycles and by Lemma~\ref{fouqqut}, there exists, in such $2$-factor, a set $A$ of type I such that $G^{2}[A]$ is an empty graph.
Moreover, by virtue of Proposition~\ref{K5} and by Brooks' theorem, every set $B$ of type II in a $2$-factor of a cubic graph of order at least 12 satisfies $\chi(G^{2}[B])\le 4$.
Consequently, by Lemma~\ref{lem1111}.ii) and by computation on the cubic graphs of order up to 20, we obtain the following result for the sequences containing two times the integer 1 and a bounded number of times the integer 2.
\begin{cor}
Every cubic graph having a 2-factor is $(1,1,2^5)$-colorable. Every $3$-edge-colorable cubic graph having a $2$-factor is $(1,1,2^4)$-colorable. 
\end{cor}

We now present sharper results for graphs with no short cycles. Note that the graphs considered in the two following results are the subcubic graphs. By Euler's formula, there do not exist finite planar cubic graphs of girth at least 6.
\begin{lem}\label{lem7}
Let $G$ be a subcubic graph having $2$-factor $\mathcal{F}$ and let $B\subseteq E(\mathcal{F})$ be of type II. If $G$ is planar and of girth at least 7, then $G^{2}[B]$ has no triangle and is planar.
\end{lem}
\begin{proof}Let $\mathcal{F}$ be any $2$-factor of $G$ and let $B\subseteq E(\mathcal{F})$ be a set of type II.
Suppose that $e_1=x_1y_1,e_2=x_2y_2,e_3=x_3y_3$ are three edges of $B$ that form a triangle in $G^{2}[B]$.  Then $e_1$ is at distance 2 of $e_2$ and hence one extremity of $e_1$, say $x_1$, is adjacent with one extremity of $e_2$, say $x_2$, i.e., $x_1x_2\in E(G)$. Similarly, $y_1$ is adjacent with one extremity of $e_3$, say $y_3$ and $x_3$ is adjacent with $y_2$. But then $x_1,x_2,y_2,x_3,y_3,y_1,x_1$ is a cycle of length six, contradicting the hypothesis on the girth. Consequently, $G^{2}[B]$ has no triangles.
The fact that $G^{2}[B]$ is planar can be seen by drawing it 'on' a cross-free embedding of the graph $G$ on the plane, putting vertices of $G^{2}[B]$ on the corresponding edges of $G$ and drawing edges of $G^{2}[B]$ along the shortest paths between corresponding edges of $G$.
\end{proof}

\begin{prop}\label{g7}
If a subcubic graph $G$ of girth at least $7$ is a subgraph of a planar cubic graph $G'$ having a $2$-factor, then $G$ is $(1,1,2^4)$-colorable. Moreover, if $G'$ is bridgeless, then $G$ is $(1,1,2^{3})$-colorable.
\end{prop}
\begin{proof}Let $G$ be a subcubic graph of girth at least $7$ and let $G'$ be a planar cubic graph having a $2$-factor such that $G$ is a subgraph of $G'$. Let $\mathcal{F}$ be a $2$-factor of $G'$ containing a minimum number of odd cycles. By Lemma~\ref{fouqqut}, there exists in $\mathcal{F}$ a set $A$ of type I such that $G'^{2}[A]$ is an empty graph.

Tait~\cite{Tait} has proven that the four color theorem is equivalent to the statement that no snark is planar. Thus, by the four color theorem, $G'$ is $3$-edge-colorable in the case $G'$ is bridgeless and the $2$-factor $\mathcal{F}$ of $G'$ (containing a minimum number of odd cycles) contains no odd cycles.

If $G'$ contains a bridge (contains no bridges, respectively), let $B$ be a set of type II in $\mathcal{F}-A$ (in $\mathcal{F}$, respectively). Let $B'=B\cap E(G)$. Thus, by Lemma~\ref{lem7}, $G^{2}[B']$ has no triangle and is planar. By the famous Grötzsch's theorem, $G^{2}[B']$ is 3-colorable. Finally, by Lemma~\ref{lem1111}.ii) $G$ is $(1,1,2^{4})$-colorable and even $(1,1,2^{3})$-colorable in the case $G'$ is bridgeless.
\end{proof}

We now show necessary conditions for a cubic graph to be $(1,1,2,2)$-colorable. For this, we first show a relation between $(1,1,2,2)$-coloring and a coloring of the vertices with two colors such that the subgraph induced by each color class has degree exactly one. This relation is used to prove Proposition \ref{neccond}.
A graph $G$ is said to be {\em 2-matching-colorable} if there exists a partition of $V(G)$ into two sets $A_1$ and $A_2$ such that both $G[A_1]$ and $G[A_2]$ are graphs of minimum and maximum degree $1$, i.e., matchings.
Note that a $2$-matching colorable graph is also what is called, in the context of defective coloring, a $(2,1)$-colorable graph \cite{CO1986}, where the list $(2,1)$ means here that we can use two colors for which each of the color class induce a subgraph of maximum degree 1. Planar graphs have been especially studied in the context of defective coloring \cite{CO1986,WG1,Poh}. 

\begin{prop}\label{mamere}
Let $G$ be a cubic graph.
The graph $G$ is $2$-matching-colorable if and only if $G$ is $(1,1,2,2)$-colorable.
\end{prop}
\begin{proof}
Suppose $G$ is $2$-matching-colorable in two set $A_1$ and $A_2$. Note that the edges of $G[A_1]$ form a color class of radius $2$ in $G$ and the same goes for $G[A_2]$. Moreover, the edges which are not in $G[A_1]$ or $G[A_2]$ form a disjoint union of even cycles. Thus, we can easily color these edges with two colors of radius $1$.
Therefore, $G$ is $(1,1,2,2)$-colorable.

Now suppose that $G$ is $(1,1,2,2)$-colorable. Let $X_1$ be the set of vertices incident with an edge colored with the first color of radius $2$ and let $X_2$ be the set of vertices incident with an edge colored with the second color of radius $2$. Note that, by definition of color of radius $2$, both $X_1$ and $X_2$ should induce a graph of maximum degree $1$.

We now prove that $X_1\cap X_2= \emptyset$ and, afterward, that $X_1\cup X_2= V(G)$.
First, suppose that there exists a vertex $u$ in $X_1\cap X_2$. Let $e$ be the edge incident with $u$ which does not have an extremity in $X_1$ or $X_2$ and let $v$ be the other extremity of $e$. By hypothesis, $e$ should be colored with a color of radius $1$. The two edges incident with $v$ (different from $e$) can not be both colored (note that we can not use any color of radius 2 for these two edges) and we obtain that $G$ is not $(1,1,2,2)$-colorable.

Second, suppose that there exists a vertex $u$ not in $X_1\cup X_2$. Since $u$ is incident with three edges which do not have colors of radius $2$, we obtain a contradiction with the fact that $G$ is $(1,1,2,2)$-colorable.

Therefore, we obtain that $X_1$ and $X_2$ form a partition of $V(G)$ and that $G$ is $2$-matching colorable.
\end{proof}

\begin{prop}\label{neccond}
Any $(1,1,2,2)$-colorable cubic graph $G$ satisfies the following properties:
\begin{enumerate}
\item[i)] $G$ is $3$-edge colorable;
\item[ii)] $G$ has order divisible by four.
\end{enumerate}
\end{prop}
\begin{proof}
By Proposition \ref{mamere}, $G$ is $2$-matching-colorable in two set $A_1$ and $A_2$.\newline
i) Note that the edges of $G[A_1]\cup G[A_2]$ form a perfect matching. Note also that the edges with one extremity in $A_1$ and the other extremity in $A_2$ form a disjoint union of even cycles. Consequently, $G$ is $3$-edge-colorable.\newline
ii) Since $G$ is cubic, we have $|A_1|=|A_2|$. Moreover, by definition of $2$-matching colorable, $|A_1|$ is even. Therefore, $G$ has order divisible by four.
\end{proof}

We end this subsection by pointing out that there exist cubic graphs that are not $(1,1,2,2,2)$-colorable. The smallest such graph, illustrated in Figure~\ref{non11222} has order 12 and has been found by exhaustive search using a computer.

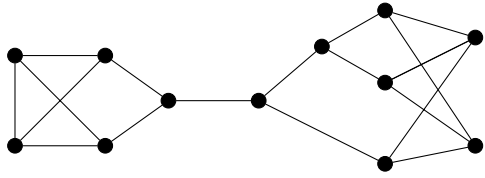
\begin{figure}
\begin{center}

\begin{tikzpicture}[scale=1.2]
\node at (0,0) [circle,draw=black,fill=black,scale=0.5](x1){};
\node at (1,0) [circle,draw=black,fill=black,scale=0.5](x2){};
\node at (1,1) [circle,draw=black,fill=black,scale=0.5](x3){};
\node at (0,1) [circle,draw=black,fill=black,scale=0.5](x4){};
\node at (1.7,0.5) [circle,draw=black,fill=black,scale=0.5](x5){};
\node at (2.7,0.5) [circle,draw=black,fill=black,scale=0.5](y1){};
\node at (4.1,-0.2) [circle,draw=black,fill=black,scale=0.5](y2){};
\node at (3.4,1.1) [circle,draw=black,fill=black,scale=0.5](y3){};
\node at (4.1,0.7) [circle,draw=black,fill=black,scale=0.5](y4){};
\node at (4.1,1.5) [circle,draw=black,fill=black,scale=0.5](y5){};
\node at (5.1,1.2) [circle,draw=black,fill=black,scale=0.5](y6){};
\node at (5.1,0) [circle,draw=black,fill=black,scale=0.5](y7){};
\draw  (x1) -- (x2) -- (x4) -- (x1) -- (x3) -- (x5) -- (x2);
\draw  (x3) -- (x4);
\draw  (x5) -- (y1) -- (y2) -- (y6) -- (y4) -- (y7);
\draw  (y1) -- (y3) -- (y5) -- (y6) -- (y4) -- (y3);
\draw  (y2) -- (y7) -- (y5);
\end{tikzpicture}

\end{center}
\caption{\label{non11222}The smallest non $(1,1,2,2,2)$-colorable and non $(1,2^6)$-colorable cubic graph.}
\end{figure}

It seems that the results of Proposition~\ref{g7} can be extended to the whole class of cubic graphs. We state this as a conjecture:
\begin{con}
 Every cubic graph is $(1,1,2,2,2,2)$-colorable and every 3-edge-colorable cubic graph is $(1,1,2,2,2)$-colorable.
\end{con}
 
\subsection{$(1,1,3,\ldots,3)$-coloring and $(1,1,4,\ldots,4)$-coloring}
We finish this section by giving general results about the required number of integers $3$ and $4$ in order that all cubic graph having a $2$-factor are $(1,1,3,\ldots,3)$-colorable and $(1,1,4,\ldots,4)$-colorable.
\begin{prop}\label{2*1+3}
Every cubic graph $G$ having a $2$-factor is $(1,1,3^{11})$-colorable. Moreover, if $G$ is $3$-edge-colorable, then $G$ is $(1,1,3^{9})$-colorable. Also there exists a $3$-edge-colorable cubic graph which is not $(1,1,3^{6})$-colorable
\end{prop}
\begin{proof}
By Proposition~\ref{akbk}, every $3$-edge-colorable cubic graph having a $2$-factor is $(1,1,3^{9})$-colorable. By Theorem~\ref{2*3}, we can color a set of type I with two colors of radius $3$. Thus, we obtain that every cubic graph having a $2$-factor is $(1,1,3^{11})$-colorable.

By exhaustive search, we have found a $3$-edge-colorable cubic graph $G$ of $14$ vertices such that the line graph of $G$ has diameter $3$, i.e., every two edges of $G$ are at distance at most $3$. Thus, since a matching of $G$ has size at most $7$, we can color at most $14$ edges with two colors of radius $1$ and we obtain that $G$ is not $(1,1,3^{6})$-colorable since there remain seven uncolored edges.
\end{proof}

\begin{prop}\label{2*1+4}
Every cubic graph $G$ having a $2$-factor is $(1,1,4^{26})$-colorable. Moreover, if $G$ is $3$-edge-colorable, then $G$ is $(1,1,4^{21})$-colorable. Also there exists a $3$-edge-colorable cubic graph which is not $(1,1,4^{14})$-colorable.
\end{prop}
\begin{proof}
By Proposition~\ref{akbk}, every $3$-edge-colorable cubic graph having a $2$-factor is $(1,1,4^{21})$-colorable. By Theorem~\ref{5*4}, we can color a set of type I with five colors of radius $4$. Thus, we obtain that every cubic graph having a $2$-factor is $(1,1,4^{26})$-colorable.

By exhaustive search, we have found a $3$-edge-colorable cubic graph $G$ of $30$ vertices ($G$ has girth $7$) such that the line graph of $G$ has diameter $4$, i.e., every two edges of $G$ are at distance at most $4$. Thus, since a matching of $G$ has size at most $15$, we can color at most $30$ edges with two colors of radius $1$ and we obtain that $G$ is not $(1,1,3^{14})$-colorable since there remain fifteen uncolored edges.
\end{proof}
\section{$(1,k,\ldots,k)$-coloring}
In this section, we show that for any cubic graph having a $2$-factor and any fixed integer $k$, there is a coloring with only one color of radius 1 and a finite number of colors of radius $k$.

\begin{theo}\label{ak2bk}
For any $k\ge 2$, every cubic graph having a $2$-factor is $(1,k^{a_{k}+2b_k+3})$-colorable and every $3$-edge-colorable cubic graph having a $2$-factor is $(1,k^{2b_k+2})$-colorable, where $a_k=\frac{2^{k+1}-(-1)^{k+1}-3}{3}$ and $b_k=\frac{2^{k+3}+2(-1)^{k+1}-6}{3}$.
\end{theo}
\begin{proof}
Let $G$ be a cubic graph and $k\ge 2$ be an integer. Let $\mathcal{F}$ be any $2$-factor of $G$.
Clearly, $E(\mathcal{F})$ can be partitioned into three sets $A$, $B$ and $C$ such that $A$ is of type I and $B,C$ are of type II.
By Lemma \ref{lem1}, $\Delta (G^{k}[A])\le a_k$. Thus, by Brooks' theorem, $G^{k}[A]$ is $(a_{k}+1)$-colorable.
By Lemma \ref{lem2}, $\Delta (G^{k}[B])\le b_k$, and $\Delta (G^{k}[C])\le b_k$. Thus, by Brooks' theorem, $G^{k}[B]$ and $G^{k}[C]$ are $(b_{k}+1)$-colorable.

Therefore, by Lemma~\ref{lem1111}.iii), $G$ is $(1,k^{2b_{k}+a_k+3})$-colorable and even $(1,k^{2b_{k}+2})$-colorable in the case $G$ is $3$-edge-colorable (by setting $A=\emptyset$).
\end{proof}

We remark that this general result is far from tight at least for small values of $k$: for $k=2$ it gives that cubic graphs having a $2$-factor are $(1,2^{9})$-colorable, but it is known since a long time that such graphs are $(2^{10})$-colorable~\cite{And}. 

For the case $k=2$ with restrictions on the graph, we can prove sharper results.
\begin{prop}
If a subcubic graph $G$ of girth at least $7$ is a subgraph of a planar cubic graph $G'$ having a $2$-factor, then $G$ is $(1,2^7)$-colorable. Moreover, if $G'$ is bridgeless, then $G$ is $(1,2^{6})$-colorable.
\end{prop}

\begin{proof}
Let $G$ be a subcubic graph of girth at least $7$ and let $G'$ be a planar cubic graph having a $2$-factor such that $G$ is a subgraph of $G'$. Let $\mathcal{F}$ be a $2$-factor of $G'$ containing a minimum number of odd cycles. By Lemma~\ref{fouqqut}, there exists in $\mathcal{F}$ a set $A$ of type I such that $G'^{2}[A]$ is an empty graph.

Tait~\cite{Tait} has proven that the four color theorem is equivalent to the statement that no snark is planar. Thus, by the four color theorem, $G$ is $3$-edge-colorable in the case $G'$ is bridgeless and the $2$-factor $\mathcal{F}$ of $G'$ (containing a minimum number of odd cycles) contains no odd cycles.

If $G'$ contains a bridge (contains no bridges, respectively), let $B$ and $C$ be two sets of type II forming a partition of $\mathcal{F}-A$ ($\mathcal{F}$, respectively). Let $B'=B\cap E(G)$ and $C'= C\cap E(G)$.
Thus, by Lemma~\ref{lem7}, both $G^{2}[B']$ and $G^{2}[C']$ have no triangle and are planar. By the famous Grötzsch's theorem, both $G^{2}[B']$ and $G^{2}[C']$ are 3-colorable. Finally, by Lemma~\ref{lem1111}.iii), $G$ is $(1,2^{7})$-colorable and even $(1,2^{6})$-colorable in the case $G'$ is bridgeless.

\end{proof}

As for the previous section, note that there exist non $(1,2^6)$-colorable cubic graphs, the smallest one being the graph on 12 vertices depicted in Figure~\ref{non11222}.

We give the following results for sequences of type $(1,3,\ldots,3)$ or $(1,4,\ldots,4)$.

\begin{prop}
Every cubic graph $G$ having a $2$-factor is $(1,3^{20})$-colorable and $(1,4^{47})$-colorable. Moreover, if $G$ is 3-edge-colorable, then $G$ is $(1,3^{18})$-colorable and $(1,4^{42})$-colorable. Also there exists two 3-edge-colorable cubic graphs $G'$ and $G''$ such that $G'$ is not $(1,3^{13})$-colorable and $G''$ is not $(1,4^{29})$-colorable.
\end{prop}
\begin{proof}
We can easily prove that $G$ is $(1,3^{20})$-colorable and even $(1,3^{18})$-colorable in the case $G$ is $3$-edge-colorable using the same arguments than in the proof of Proposition~\ref{2*1+3}. Using the same arguments than in the proof of Proposition~\ref{2*1+4}, we can prove analogous results for sequences of type $(1,4,\ldots,4)$.

The graph $G'$ is the non $(1,1,3^{6})$-colorable cubic graph of Proposition~\ref{2*1+3} and $G''$ is the non $(1,1,4^{14})$-colorable cubic graph of Proposition~\ref{2*1+4}. Since the line graph of $G'$ has diameter $3$ and the line graph of $G''$ has diameter 4 and since the matchings of $G'$ has size at most $7$ and the matching of $G''$ has size at most $15$, we obtain that $G'$ is not $(1,3^{13})$-colorable and that $G''$ is not $(1,4^{29})$-colorable.

\end{proof}
The computations we have made let us think that seven colors of radius 2 are enough in general and less with girth restrictions.
We end this section by stating these two open problems.
\begin{quest}
 Is it true that all cubic graphs are $(1,2^7)$-colorable ?
\end{quest}

\begin{quest}
 Is it true that all cubic graphs of girth at least 5 are $(1,2^5)$-colorable ?
\end{quest}

\section{$(1,2,\ldots,k)$-coloring}
We finish this paper by proving that there is no integer $k$ such that every subcubic graph is $(1,2,\ldots,k)$-colorable, i.e., that the line graphs of subcubic graphs have arbitrary large packing chromatic number. 
We recall that the packing chromatic number of a graph $G$ is the smallest integer $k$ such that there exists a partition of $V(G)$ into $k$ subsets $\{X_{1},\ldots, X_{k}\}$, each $X_i$ being a set of vertices at pairwise distance at least $i+1$. Note that it has already been proven that for every fixed $k$ and $g\ge2k+2$, almost every cubic graph of girth at least $g$ of sufficiently large order has packing chromatic number greater than $k$~\cite{Bal17}.

Adding a leaf on a vertex $u$ of a graph $G$ is an operation that consists in adding a new vertex $v$ and the edge $uv$.

Let $T_1$ be the graph $K_{1,3}$ and let $T'_1$ be the graph $K_{1,3}$ for which we have added two leaves on the same vertex of degree $1$. By induction on $i$, let $T_i$ ($T'_i$, respectively), for $i\ge2$, be the graph constructed from $T_{i-1}$ ($T'_{i-1}$, respectively) by adding two leaves on each vertex $u$ of degree $1$.

\begin{prop}\label{trigraph}
For any integer $k$, there exists an integer $N$ such that $T_N$ is not $(1,2,\ldots,k)$-colorable.
\end{prop}
\begin{proof}
By induction on $i$, we begin by proving that the diameter of the line graph of $T_i$ is $2i-1$.
For $i=1$, the diameter of the line graph of $T_1$ is $1$. Suppose that the diameter of the line graph of $T_i$ is $2i-1$. By definition, we obtain that the diameter of the line graph $T_{i+1}$ is $2i+1$. Similarly, we can prove that the diameter of the line graph of $T'_i$ is $2i$.

Moreover, by induction on $i$, we prove that $|E(T_i)|=3 (2^{i}-1)$.
For $i=1$, $|E(T_i)|=3$. Suppose that $|E(T_i)|=3 (2^{i}-1)$. It is trivial to note that there are $3 (2^{i}-1)- 3 (2^{i-1}-1)=3\ 2^{i-1}$ vertices of degree $1$ in $T_i$. Thus, there are $2 (3\ 2^{i-1})$ new edges in $T_{i+1}$. Consequently, $|E(T_{i+1})|=3 (2^{i}-1) +2(3 \ 2^{i-1})=3(2^{i+1}-1)$.

Also, by induction on $i$, we prove that $|E(T'_i)|=2^{i+2}-3$.
For $i=1$, $|E(T'_i)|=5$. Suppose that $|E(T'_i)|=2^{i+2}-3$. It is trivial to note that there are $2^{i+2}-3 -2^{i+1}+3=2^{i+1}$ vertices of degree $1$ in $T'_i$. Thus, there are $2 (2^{i+1})$ new edges in $T'_{i+1}$. Consequently, $|E(T'_{i+1})|=2^{i+2}-3+ 2( 2^{i+1})=2^{i+3}-3$.

Let $N$ be a sufficiently large integer (compared to $k$). Since the diameter of the line graph of $T_i$ is $2i-1$, an upper bound on the size of an $i$-packing (a set of edges at pairwise distance at least $i$) in $T_{N}$, for $i$ an odd integer, converges towards $|E(T_{N})|/|E(T_{(i+1)/2})|$.
Moreover, since the diameter of the line graph of $T'_i$ is $2i$, an upper bound on the size of an $i$-packing in $T_{N}$, for $i$ an even integer, converges towards $|E(T_{N})|/|E(T'_{i/2})|$.
Thus, if $T_{N}$ is $(1,2,\ldots,k)$-colorable and $\epsilon$ is an arbitrary small constant, then $$\sum^{\lceil k/2 \rceil}_{i=1}(|E(T_{N})|/|E(T_{i})| ) +\sum^{\lfloor k/2 \rfloor}_{i=1}(|E(T_{N})|/|E(T'_{i})| )-\epsilon \ge |E(T_{n})|. $$

However, by calculation, $$\sum^{\lceil k/2 \rceil}_{i=1}(1/|E(T_{i})| ) +\sum^{\lfloor k/2 \rfloor}_{i=1}(1/|E(T'_{i})|)\le \sum_{i=1}^{\infty} 1/(3 (2^{i}-1))+ \sum_{i=1}^{\infty} 1/(2^{i+2}-3)<0.8793<1.$$ Thus, we obtain a contradiction and $T_{N}$ is not $(1,2,\ldots,k)$-colorable.

\end{proof}
 Since $T_N$ is subcubic for any integer $N$ and since for any integer $n>N$, $T_n$ contains $T_N$ as subgraph, we obtain the following corollary.
\begin{cor}
There exist non $(1,2,\ldots,k)$-colorable subcubic graphs of arbitrary large order for every integer $k$.
\end{cor}
Note also that there exists non $(1,2,\ldots,k)$-colorable cubic graphs for every integer $k$, since we can easily construct a cubic graph containing $T_{N}$ as subgraph, for every integer $N$.
\section*{Acknowledgments} 
We thank Borut Lu\v{z}ar for pointing out references~\cite{Fou2,Pay}.


\begin{thebibliography}{10}
\bibitem{And} L. D. Andersen, The strong chromatic index of a cubic graph is at most 10, {\em Discrete Mathematics}, 108(1): 231--252, 1992.


\bibitem{Bal17} J. Balogh, A. Kostochka, X. Liu,
Packing chromatic number of cubic graphs, 
{\em Discrete Mathematics}, 341(2): 474--483, 2018.


\bibitem{BresKlav}
B.~ Bre{\v{s}}ar, S.~Klav{\v{z}}ar, D.~F. Rall,
Packing chromatic number of base-3 sierpiński graphs,
{\em Graphs and Combinatorics}, 32(4): 1313--1327, 2016.

\bibitem{BKRW}
B. Brešar, S. Klavžar, D. F. Rall, K. Wash, 
Packing chromatic number under local changes in a graph, 
{\em Discrete Mathematics}, 340(5): 1110--1115, 2017.

\bibitem{BGHM}
G. Brinkmann, J. Goedgebeur, J. Hägglund, K. Markström, 
Generation and properties of snarks, 
{\em Journal of Combinatorial Theory}, Series B, 103(4): 468--488, 2013.

\bibitem{CO1986}
R. Cowen, R. H. Cowen, D. R. Woodall,
Defective colorings of graphs in surfaces: Partitions into subgraphs of bounded valency,
{\em Journal of Graph Theory } 10 (2): 187--195, 1986.

\bibitem{FiCo}
J. Fiala, P. A. Golovach,
Complexity of the packing coloring problem for trees,
{\em Discrete Appl. Math.} 158(7): 771--778, 2010.

\bibitem{Fin} A. S. Finbow, D. F. Rall, On the packing chromatic number of some lattices, {\em Discrete Appl. Math.} 158(12): 1224--1228, 2010.


\bibitem{Fou}J. L. Fouquet, J. L. Jolivet,
Strong edge-colorings of graphs and applications to multi-k-gons,
{\em Ars Combinatoria}, 16A: 141--150, 1983.

\bibitem{Fou3}J. L. Fouquet, J. L. Jolivet,
Tools for parsimonious edge-coloring of graphs with maximum degree three,
{\em on arXiv}: 1201.5988v1, 2012.

\bibitem{Fou2}J. L. Fouquet, J. M. Vanherpe,
On parsimonious edge-colouring of graphs with maximum degree three,
{\em Graphs and Combinatorics}, 29(3): 475-487, 2013.


\bibitem{Gasto} N.~Gastineau, O.~Togni,
$S$-packing colorings of cubic fraphs, 
{\em Discrete Mathematics} 339(10): 2461--2470, 2016.

\bibitem{nico}
N. Gastineau, H. Kheddouci, O. Togni,
Subdivision into $i$-packings and $S$-packing chromatic number of some lattices,
\emph{Ars Mathematica Contemporanea} 9:331--354, 2015.

\bibitem{GHT}
N. Gastineau, P. Holub, O. Togni,
On the packing chromatic number of subcubic outerplanar graphs,
{\em arXiv}:1703.05023 [cs.DM], 2017.

\bibitem{WG1}
W. Goddard, Acyclic colorings of planar graphs,
{\em Discrete Mathematics}, 91(1): 91--94, 1991.

\bibitem{GX1}
W. Goddard, H. Xu, The $S$-packing chromatic number of a graph, 
{\em Discuss. Math. Graph Theory}, 32: 795--806, 2012.

\bibitem{GX2}
W. Goddard, H. Xu, A note on $S$-packing colorings of lattices, 
{\em Discrete Appl. Math.}, 166: 255--262, 2014.

\bibitem{God} W.~Goddard, S.~M.~Hedetniemi, S.~T.~Hedetniemi, J.~M.~Harris, D.~F.~Rall,
Broadcast chromatic numbers of graphs, {\em Ars Combinatoria}, 86: 33--49, 2008.

\bibitem{Hoc}
H. Hocquard, M. Montassier, A. Raspaud, P. Valicov, 
On strong edge-colouring of subcubic graphs, {\em Discrete Appl. Math.}, 161(16): 2467--2479, 2013.


\bibitem{Luk}
R. Lukoťka, E. Máčajová, J. Mazák, M. Škoviera,
Small snarks with large oddness, {\em The Electronic Journal of Combinatorics}, 22(1): \#P1.51, 2015.


\bibitem{Pay}
C. Payan,
Sur quelques problèmes de couvertures et de couplages en combinatoire. Thèse d'état (in french), 1977.

\bibitem{Poh}
K. S. Poh,
On the linear vertex-arboricity of a planar graph, {\em Journal of Graph Theory}, 14(1):73--75, 1990.
\bibitem{Tait}
P. G. Tait,
Remarks on the colourings of maps,
{\em Proc. R. Soc. Edinburgh}, 10:729, 1880.
\end{thebibliography}
\end{document}